\documentclass[11pt]{article}

\usepackage[a4paper,margin=1in]{geometry}

\usepackage[T1]{fontenc}
\usepackage{lmodern}
\usepackage{abstract}
\usepackage{array}
\usepackage{amsmath}
\usepackage{amssymb}
\usepackage{amsthm}
\usepackage{mathtools}
\usepackage{mathrsfs}
\usepackage{stmaryrd}  

\usepackage{enumitem}
\usepackage{booktabs}
\usepackage{multirow}
\usepackage{float}
\usepackage{pifont}
\usepackage{xcolor}

\usepackage{tikz}
\usetikzlibrary{positioning,arrows.meta,shapes,calc,decorations.pathreplacing,fit}

\usepackage{url}
\usepackage{hyperref}
\hypersetup{
	colorlinks=true,
	linkcolor=blue,
	citecolor=blue,
	urlcolor=blue
}

\newcommand{\Prop}{\mathsf{Prop}}

\newcommand{\coalc}{\overline{C}}

\newcommand{\sem}[1]{\ensuremath{\llbracket {#1} \rrbracket}}

\newcommand{\CL}{\mathsf{CL}}
\newcommand{\ATL}{\mathsf{ATL}}

 \newcommand{\CLFI}{\mathsf{CL}^{\FI}}

\newcommand{\Eff}[1]{\langle #1\rangle}

\newcommand{\FC}{\mathsf{FC}}
\newcommand{\PD}{\mathsf{PD}}
\newcommand{\AD}{\mathsf{AD}}
\newcommand{\FI}{\mathsf{FI}}

\newcommand{\powerset}{\mathcal{P}}
\newcommand{\defeq}{\coloneqq}

\newtheorem{observation}{Observation}

\newcommand{\PSPACE}{\mathbf{PSPACE}}

\newcommand{\Lang}{\mathcal{L}}
\newcommand{\Mod}{\mathcal{M}}

\newcommand{\tr}{t}

\theoremstyle{definition}
\newtheorem{definition}{Definition}[section]
\newtheorem{example}[definition]{Example}

\theoremstyle{plain}
\newtheorem{theorem}[definition]{Theorem}
\newtheorem{lemma}[definition]{Lemma}
\newtheorem{proposition}[definition]{Proposition}
\newtheorem{corollary}[definition]{Corollary}

\theoremstyle{remark}
\newtheorem{remark}[definition]{Remark}

\newcommand{\Rfc}{\mathcal{R}_{\mathsf{FC}}}
\newcommand{\Rpd}{\mathcal{R}_{\mathsf{PD}}}
\newcommand{\Rad}{\mathcal{R}_{\mathsf{AD}}}
\newcommand{\Rfi}{\mathcal{R}_{\mathsf{FI}}}

\newcommand{\coop}[1]{\langle\!\langle #1 \rangle\!\rangle}

\newenvironment{keywords}
{\par\noindent\textbf{Keywords: }}
{\par\medskip}

\title{\textbf{Beyond Ability: The Four-Fold Spectrum of Power and the Logic of Full Inability}}

\author{
	Shanxia Wang\\
	School of Computer and Information Engineering
	(School of Artificial Intelligence)\\
	Henan Normal University, Xinxiang City,
	Henan Province, China\\
	\texttt{wangshanxia@htu.edu.cn}
}
\date{}

\begin{document}
	
	\maketitle
	
	\begin{abstract}
		Coalition Logic studies what coalitions can enforce. Recent work treats inability as simple non-ability: $\neg\Eff{C}\varphi$. This conflates two distinct configurations---a coalition unable to force $\varphi$ may still force $\neg\varphi$, retaining adversarial control rather than genuine inability. We introduce \textbf{Full Inability} ($\FI$): the symmetric condition in which a coalition can enforce neither a proposition nor its negation.
		
		Combining coalitional effectivity with propositional negation yields a four-fold spectrum: \textbf{Full Control} ($\FC$), \textbf{Positive Determination} ($\PD$), \textbf{Adverse Determination} ($\AD$), and \textbf{Full Inability} ($\FI$). These categories partition a coalition's strategic status exhaustively and exclusively. We establish their algebraic and order-theoretic structure. Under $\alpha$-duality, propositional negation and coalition complementation generate a Klein four-group symmetry. In playable models, the four power regions are order-convex in the powerset lattice, yielding interval-stable verification of inability.
		
		We axiomatize $\CLFI$, a definitional extension treating Full Inability as a primitive modality. Via elimination translation, we prove soundness, completeness, and conservativity over Coalition Logic. The extension preserves expressive power and complexity ($\PSPACE$-complete), but provides direct proof-theoretic access to symmetric inability, strategic dependence, propositional dummyhood, and containment verification.
	\end{abstract}
	

	\section{Introduction}
	\label{sec:introduction}
	
	\subsection{The Logical Structure of Strategic Inability}
	\label{subsec:blind-spot}
	
	Since Pauly's seminal work \cite{Pauly02}, Coalition Logic ($\CL$) has become the standard framework for reasoning about the strategic abilities of groups of agents \cite{Wooldridge09}. Its central modality $\Eff{C}\varphi$ expresses that coalition $C$ can guarantee $\varphi$. The paradigm has been extended to temporal dynamics \cite{AHK02}, epistemic uncertainty, and resource constraints. Yet the negative dimension of strategic power---what coalitions \emph{cannot} determine---has long remained derivative.
	
	Recent work \cite{wang2026logic} elevated inability to an independent modality, introducing the operator $\text{Iab}_C\varphi \equiv \neg\Eff{C}\varphi$ and establishing its structural properties: anti-monotonicity with respect to coalition inclusion, contravariance with respect to goal strength, and asymmetric interaction with Boolean connectives. This treatment made inability a first-class concept, but it remained a simple negation of ability---a binary distinction between what a coalition can and cannot enforce.
	
	\subsection{From Simple Inability to Full Inability}
	\label{subsec:from-simple-to-full}
	
	The binary treatment of inability is too coarse for fine-grained strategic analysis. A coalition unable to force $\varphi$ need not be powerless: it may still force $\neg\varphi$, thereby possessing \emph{adversarial control} rather than genuine strategic neutralization. The core limitation of equating inability with mere non-ability is that it conflates two structurally distinct configurations:
	\begin{enumerate}[label=(\roman*)]
		\item the coalition cannot force $\varphi$ but can force $\neg\varphi$---a state we term \textbf{Adverse Determination} ($\AD$);
		\item the coalition can force neither $\varphi$ nor $\neg\varphi$---the state of \textbf{Full Inability} ($\FI$).
	\end{enumerate}
	
	To separate these cases we introduce the symmetric notion of \textbf{Full Inability}:
	\[
	\FI_C(\varphi) \;\defeq\; \neg\Eff{C}\varphi \;\wedge\; \neg\Eff{C}\neg\varphi.
	\]
	Thus $\FI_C(\varphi)$ expresses a symmetric absence of deterministic power. From $C$'s perspective the truth value of $\varphi$ remains strategically unsettled: its resolution depends on the environment, the complementary coalition $\coalc$, or their interaction.
	
	By simultaneously evaluating a coalition's ability to enforce both $\varphi$ and $\neg\varphi$, we obtain a four-fold spectrum of strategic power. This spectrum reveals algebraic symmetries and order-theoretic structure that remain invisible when inability is treated as a single undifferentiated concept. Shifting from a binary contrast to a four-fold classification supplies a precise instrument for isolating situations in which a coalition is completely neutralized with respect to a given proposition.
	
	\subsection{The Four-Fold Spectrum and Its Structure}
	\label{subsec:four-fold-overview}
	
	The positive counterpart of Full Inability is \textbf{Full Control} ($\FC$), which holds when a coalition can enforce either truth value:
	\[
	\FC_C(\varphi) \;\defeq\; \Eff{C}\varphi \;\wedge\; \Eff{C}\neg\varphi.
	\]
	In arbitrary playable models these two extremes satisfy only one-way polarity: $\FC_C(\varphi)$ entails $\FI_{\coalc}(\varphi)$, but the converse fails in general. Exact dual equivalence emerges only under the stronger assumption of $\alpha$-duality (Section~\ref{sec:duality-theory}).
	
	Together with the two one-sided cases, $\FI$ and $\FC$ generate the four-fold spectrum
	\[
	\FC \quad \PD \quad \AD \quad \FI,
	\]
	standing for Full Control, Positive Determination, Adverse Determination, and Full Inability. The partition is reminiscent of Belnap's four-valued logic \cite{Belnap77}, yet its interpretation is coalitional rather than semantic: the four values classify strategic power over a proposition.
	
	\subsection{Motivating Scenarios}
	\label{subsec:motivating-scenarios}
	
	The need for $\FI$ as a refinement of inability appears across several domains.
	
	\paragraph{Social Choice and Voting.}
	In weighted voting games \cite{BCG16} a \emph{dummy player} is one whose inclusion never alters a coalition's winning status. $\FI$ captures a local, proposition-sensitive form of non-pivotality. A marginalized voter may satisfy $\FI$ with respect to a bill: acting alone the voter can neither pass nor veto it. Simple inability such as $\neg\Eff{\{i\}}\mathsf{pass}$ fails to distinguish this voter from one who still retains veto power.
	
	\paragraph{Game Theory.}
	In the classic Matching Pennies game let $p$ denote that the coins match. Each individual player $i$ satisfies $\FI_{\{i\}}(p)$: neither can independently guarantee a match or a mismatch. Yet the grand coalition satisfies $\FC_N(p)$. Here $\FI$ captures strategic dependence: unilateral agents lack deterministic control, while coordination restores it.
	
	\paragraph{AI Safety and Containment.}
	Containment architectures \cite{Bostrom14,Amodei16} aim to restrict the strategic reach of artificial agents. A robust sandbox may require an AI subsystem $C$ to be fully unable to determine critical safety variables. Requiring only $\neg\Eff{C}\mathsf{breach}$ is insufficient: the system might still force $\neg\mathsf{breach}$ and exploit that ability as strategic leverage. $\FI$ supplies a formal criterion for complete strategic neutralization.
	
	\subsection{Contributions and Organization}
	\label{subsec:contributions}
	
	This paper integrates Full Inability into Coalition Logic. The main contributions are:
	
	\begin{enumerate}[label=(\arabic*)]
		\item \textbf{A Four-Fold Spectrum of Power.}
		We define an exhaustive classification of coalitional status:
		\[
		\FC,\quad \PD,\quad \AD,\quad \FI.
		\]
		These categories distinguish Full Control, one-sided determination, and Full Inability (Section~\ref{sec:classification}).
		
		\item \textbf{Polarity, Duality, and Symmetry.}
		Playable models validate the one-way polarity
		\[
		\FC_C(\varphi) \Rightarrow \FI_{\coalc}(\varphi).
		\]
		Under $\alpha$-duality this becomes a dual equivalence governed by a Klein four-group symmetry (Section~\ref{sec:duality-theory}).
		
		\item \textbf{Lattice-Theoretic Convexity.}
		By mapping formulas to truth sets in the powerset lattice we prove that each power region is order-convex. This establishes the stability of Full Inability and supports interval-based verification (Section~\ref{sec:structural-properties}).
		
		\item \textbf{Axiomatization of $\CLFI$.}
		We introduce $\CLFI$, a definitional extension of Coalition Logic that treats Full Inability as a primitive modality. Via an elimination translation we establish soundness, completeness, and conservativity (Section~\ref{sec:axiomatization}).
	\end{enumerate}
	
	The paper proceeds as follows. Section~\ref{sec:related-work} surveys related work. Section~\ref{sec:preliminaries} reviews Coalition Logic and playable effectivity functions. Sections~\ref{sec:classification}--\ref{sec:structural-properties} develop the four-fold spectrum and its structural properties. Section~\ref{sec:axiomatization} presents the proof system and meta-theory of $\CLFI$. Section~\ref{sec:conclusion} concludes with future directions.

	\section{Related Work}
	\label{sec:related-work}
	
	The logical study of strategic power has developed from effectivity-based accounts of coalitional ability into a broad family of systems incorporating time, knowledge, resources, and institutional constraints \cite{Wooldridge09, Jamroga23}. This paper contributes to this line by isolating a symmetric form of coalitional non-determination---\textbf{Full Inability}---and studying its logical, algebraic, and proof-theoretic structure.
	
	\paragraph{Coalitional Ability and Effectivity.}
	Coalition Logic ($\CL$) \cite{Pauly02, Pauly01} provides a one-step modal logic of coalitional effectivity: $\Eff{C}\varphi$ states that coalition $C$ can enforce $\varphi$. Alternating-time Temporal Logic ($\ATL$) \cite{AHK02} extends this to temporal objectives over concurrent game structures. Strategic reasoning has been further developed through strategy logic \cite{MMPV14} and verification frameworks \cite{CLMR23, GHL22}. Semantic foundations for multi-agent logics have been systematically compared \cite{GJ04}. In these traditions, inability is represented by external negation, $\neg\Eff{C}\varphi$. This suffices for expressing failure to enforce a formula, but it does not distinguish one-sided adversarial control from complete strategic non-determination. We make this distinction explicit by separating Adverse Determination from Full Inability.
	
	\paragraph{Inability and Negative Agency.}
	The analysis of what agents are unable to do has a longer background in philosophical logic and theories of agency \cite{Kenny75, Horty01, Maier22}. Belnap and colleagues developed a stit-theoretic account of agency and choice in indeterminist settings \cite{BPX01}, where inability emerges from the structure of branching time. In formal multi-agent systems, inability often appears indirectly---as a consequence of imperfect information, insufficient resources, or environmental restrictions.
	
	More recently, inability has been studied as a modality in its own right \cite{wang2026logic}. That work introduced the operator $\text{Iab}_C\varphi \equiv \neg\Eff{C}\varphi$ as a primitive modality, established its axiomatization as a conservative extension of Coalition Logic, and derived its structural laws: anti-monotonicity over coalitions ($C \subseteq D \Rightarrow \text{Iab}_D\varphi \to \text{Iab}_C\varphi$), contravariance over goals ($\varphi \to \psi \Rightarrow \text{Iab}_C\psi \to \text{Iab}_C\varphi$), and asymmetric distribution over Boolean connectives. However, it treated inability as a single, undifferentiated concept---the simple negation of ability.
	
	The present paper advances this line of investigation by distinguishing \emph{simple inability} ($\neg\Eff{C}\varphi$) from \emph{Full Inability}:
	\[
	\FI_C(\varphi) \equiv \neg\Eff{C}\varphi \wedge \neg\Eff{C}\neg\varphi.
	\]
	This refinement yields a four-fold classification of coalitional status: Full Control, Positive Determination, Adverse Determination, and Full Inability. Under $\alpha$-duality, these categories exhibit a Klein four-group symmetry generated by propositional negation and coalition complementation. In playable models, they correspond to order-convex regions in the powerset lattice. Thus Full Inability is not merely a stronger form of inability, but the cornerstone of a systematic algebraic and order-theoretic structure governing strategic power. Where the earlier work established inability as a first-class modality, the present work reveals its internal structure.
	
	\paragraph{Power Indices and Social Choice.}
	Social choice theory and cooperative game theory study power through notions such as pivotality, dummy players, and quantitative power indices, including the Banzhaf and Shapley--Shubik indices \cite{BCG16, Kurz22}. Classical results on voting manipulation \cite{Gibbard73} and recent complexity-theoretic analyses \cite{EHL23} reveal the computational structure of strategic influence in voting. These approaches typically measure an agent's marginal influence globally across coalitions or voting configurations. The $\CLFI$ framework is complementary: it provides a qualitative, state-dependent, and proposition-relative analysis of strategic influence. In particular, Full Inability supports a local notion of dummyhood: an agent or coalition may be unable to determine a particular proposition at a given state even if it is not globally powerless.
	
	\paragraph{Resources, Constraints, and Strategic Neutralization.}
	A substantial body of work studies how strategic ability changes under resource bounds, action restrictions, memory limitations, or other constraints \cite{ADL14, ABL17, GLP24, BMM23}. These approaches typically restrict the strategies or capacities available to agents. Full Inability instead specifies a target condition on the resulting effectivity profile: a coalition is fully unable with respect to $\varphi$ exactly when it can enforce neither $\varphi$ nor $\neg\varphi$. This distinction is relevant to formal discussions of containment and AI safety \cite{Bostrom14, Amodei16}. When the design goal is strategic neutralization rather than beneficial control, one may require $\FI_C(\varphi)$ for critical formulas $\varphi$. Such a requirement is stronger than merely blocking one harmful outcome, since it also rules out unilateral control over the opposite condition.
	
	\paragraph{Knowledge and Strategic Ability.}
	The interaction between knowledge and ability has been extensively studied in epistemic and strategic logics \cite{FHMV95, vDHK07, AA19, AJ22}. Ruan and colleagues explored how agents reason about their own abilities in coalitional games \cite{Ruan22}. Much of this work focuses on what agents know they can achieve, or on how imperfect information affects strategic ability. The present paper does not add epistemic operators to $\CLFI$, but the primitive availability of $\FI$ suggests natural extensions. For example, $K_i\FI_C(\varphi)$ would express that agent $i$ knows coalition $C$ to be fully unable to determine $\varphi$. This would allow reasoning about acknowledged dependence, publicly recognized non-pivotality, and common knowledge of strategic neutralization. We leave such epistemic extensions for future work.
	
	\paragraph{Power in Social Networks and Games.}
	Recent work has examined power structures in social networks through logical frameworks \cite{ABG24, Liu23} and game-theoretic perspectives \cite{Turrini17, Goranko13}. Goranko, Jamroga, and Turrini established the connection between strategic games and truly playable effectivity functions \cite{Goranko13}, providing the semantic foundation for our analysis. These approaches study how network topology and strategic interaction shape coalitional influence. Our four-fold spectrum provides a complementary local analysis: at each state, we classify a coalition's status with respect to a specific proposition, revealing fine-grained patterns of power and dependence.
	
	\paragraph{Extensions and Variants of Coalition Logic.}
	Coalition Logic has been extended in several directions. First-order variants incorporate quantification over agents and propositions \cite{Galimullin25}. Minimal coalition logics explore the core expressive power of coalitional modalities \cite{Li25}. Generalized coalition logics study alternative semantic frameworks and model equivalences \cite{Chen25}. Description-logic-based approaches integrate coalitional reasoning with ontological knowledge representation \cite{Seylan09, Seylan10}. Strategy logic with communicative actions models explicit communication among agents \cite{Mittelmann24}. Our contribution is orthogonal to these extensions: Full Inability can be integrated into any of these frameworks by refining the treatment of negative coalitional power.
	
	\paragraph{Order-Theoretic and Algebraic Perspectives.}
	Effectivity semantics is intrinsically order-theoretic: each coalition is associated with a family of enforceable sets of outcomes, ordered by inclusion. The powerset lattice $(\powerset(W),\subseteq)$ therefore provides a natural algebraic environment for studying ability and inability. Standard monotonicity principles for effectivity functions already reflect this order structure \cite{DP02}. Modal logic foundations \cite{Chellas80} provide the general framework for interpreting coalitional modalities. Building on this perspective, we show that the four coalitional categories induced by $\Eff{C}\varphi$ and $\Eff{C}\neg\varphi$ correspond to order-convex regions in the powerset lattice. Under the additional assumption of $\alpha$-duality, these regions also exhibit a Klein four-group symmetry (the group $V_4 \cong \mathbb{Z}_2 \times \mathbb{Z}_2$) generated by propositional negation and coalition complementation, connecting our framework to bilattice semantics \cite{Fitting91}.

	\section{Preliminaries}
	\label{sec:preliminaries}
	
	This section establishes the formal background for Coalition Logic ($\CL$) and effectivity functions, following Pauly \cite{Pauly02, Pauly01}.
	
	\subsection{Coalition Logic and Effectivity Functions}
	\label{subsec:cl-semantics}
	
	Let $N = \{1, \ldots, n\}$ be a finite set of agents and let $\Prop$ be a countable set of atomic propositions. A \emph{coalition} is any subset $C \subseteq N$. The complementary coalition is denoted by
	\[
	\coalc \defeq N \setminus C.
	\]
	For any outcome set $X \subseteq W$, we write
	\[
	\overline{X} \defeq W \setminus X
	\]
	for its complement in $W$.
	
	\begin{definition}[Language $\Lang_{\CL}$]
		The language of Coalition Logic is generated by:
		\[
		\varphi ::= p \mid \neg\varphi \mid (\varphi \wedge \psi) \mid \Eff{C}\varphi,
		\]
		where $p \in \Prop$ and $C \subseteq N$. The formula $\Eff{C}\varphi$ is read as ``coalition $C$ can ensure $\varphi$.''
	\end{definition}
	
	The semantics is given by coalition models. At each state, an effectivity function specifies which sets of outcomes each coalition can enforce.
	
	\begin{definition}[Coalition Model]
		A \emph{coalition model} is a tuple
		\[
		\Mod = (W,E,V)
		\]
		where:
		\begin{itemize}
			\item $W$ is a non-empty set of states;
			\item $E \colon W \to (\powerset(N) \to \powerset(\powerset(W)))$ assigns to each state $w \in W$ and coalition $C \subseteq N$ a family $E_w(C)$ of outcome sets enforceable by $C$ at $w$;
			\item $V \colon \Prop \to \powerset(W)$ is a valuation.
		\end{itemize}
	\end{definition}
	
	The satisfaction relation $\Mod,w \models \varphi$ is defined inductively:
	\begin{align*}
		\Mod,w &\models p
		&&\Longleftrightarrow\quad
		w \in V(p), \\
		\Mod,w &\models \neg\varphi
		&&\Longleftrightarrow\quad
		\Mod,w \not\models \varphi, \\
		\Mod,w &\models \varphi \wedge \psi
		&&\Longleftrightarrow\quad
		\Mod,w \models \varphi \text{ and } \Mod,w \models \psi, \\
		\Mod,w &\models \Eff{C}\varphi
		&&\Longleftrightarrow\quad
		\sem{\varphi}_{\Mod} \in E_w(C),
	\end{align*}
	where
	\[
	\sem{\varphi}_{\Mod}
	=
	\{u \in W \mid \Mod,u \models \varphi\}
	\]
	is the truth set of $\varphi$ in $\Mod$.
	
	\subsection{Strategic Game Forms and Playability}
	\label{subsec:playability}
	
	The effectivity function $E_w$ abstracts the strategic possibilities available at state $w$. A \emph{strategic game form} at $w$ consists of a non-empty action set $Act_i$ for each agent $i \in N$ and an outcome function
	\[
	o \colon \prod_{i \in N} Act_i \to W.
	\]
	Coalition $C$ can enforce an outcome set $X \subseteq W$ if it has a joint strategy guaranteeing that the resulting outcome lies in $X$, regardless of how the complementary coalition acts:
	\[
	X \in E_w(C)
	\quad\Longleftrightarrow\quad
	\exists s_C \in \prod_{i \in C} Act_i\;
	\forall s_{\coalc} \in \prod_{i \in \coalc} Act_i
	\colon
	o(s_C,s_{\coalc}) \in X.
	\]
	
	Pauly's representation theorem establishes that the effectivity functions induced by strategic game forms are exactly the \emph{playable} effectivity functions. We use the following standard characterization.
	
	\begin{definition}[Playability]
		\label{def:playability}
		An effectivity function
		\[
		E_w \colon \powerset(N) \to \powerset(\powerset(W))
		\]
		is \emph{playable} if it satisfies the following conditions:
		\begin{enumerate}[label=(\roman*)]
			\item \textbf{Liveness}: $\emptyset \notin E_w(C)$ for all $C \subseteq N$.
			
			\item \textbf{Safety}: $W \in E_w(C)$ for all $C \subseteq N$.
			
			\item \textbf{Outcome Monotonicity}: if $X \in E_w(C)$ and $X \subseteq Y \subseteq W$, then $Y \in E_w(C)$.
			
			\item \textbf{Superadditivity}: if $C \cap D = \emptyset$, $X \in E_w(C)$, and $Y \in E_w(D)$, then
			\[
			X \cap Y \in E_w(C \cup D).
			\]
			
			\item \textbf{$N$-Maximality}: for every $X \subseteq W$,
			\[
			X \notin E_w(\emptyset)
			\quad\Longleftrightarrow\quad
			\overline{X} \in E_w(N).
			\]
		\end{enumerate}
	\end{definition}
	
	A coalition model $\Mod=(W,E,V)$ is \emph{playable} if $E_w$ is playable for every $w \in W$. We record two basic consequences of playability.
	
	\begin{lemma}[Coalition Monotonicity]
		\label{lem:coalition-monotonicity}
		If $E_w$ is playable and $C \subseteq D$, then
		\[
		E_w(C) \subseteq E_w(D).
		\]
	\end{lemma}
	
	\begin{proof}
		Let $X \in E_w(C)$. Since $C \subseteq D$, we can write
		\[
		D = C \cup (D \setminus C),
		\]
		with $C \cap (D \setminus C)=\emptyset$. By Safety,
		\[
		W \in E_w(D \setminus C).
		\]
		By Superadditivity applied to $C$ and $D \setminus C$,
		\[
		X \cap W = X \in E_w(C \cup (D \setminus C)) = E_w(D).
		\]
		Therefore $E_w(C) \subseteq E_w(D)$.
	\end{proof}
	
	\begin{lemma}[Regularity]
		\label{lem:regularity}
		If $E_w$ is playable, then for all $C \subseteq N$ and $X \subseteq W$,
		\[
		X \in E_w(C)
		\quad\Longrightarrow\quad
		\overline{X} \notin E_w(\coalc).
		\]
	\end{lemma}
	
	\begin{proof}
		Suppose, for contradiction, that
		\[
		X \in E_w(C)
		\quad\text{and}\quad
		\overline{X} \in E_w(\coalc).
		\]
		Since $C \cap \coalc = \emptyset$, Superadditivity gives
		\[
		X \cap \overline{X} = \emptyset \in E_w(C \cup \coalc) = E_w(N).
		\]
		This contradicts Liveness, which requires $\emptyset \notin E_w(N)$. Hence
		\[
		\overline{X} \notin E_w(\coalc).
		\]
	\end{proof}
	
	\subsection{$\alpha$-Duality}
	\label{subsec:alpha-duality}
	
	Regularity provides a one-way exclusion principle: if coalition $C$ can enforce $X$, then the complementary coalition $\coalc$ cannot enforce the complement $\overline{X}$. In some settings, this one-way implication strengthens to an equivalence.
	
	\begin{definition}[$\alpha$-Duality]
		\label{def:alpha-duality}
		An effectivity function $E_w$ satisfies \emph{$\alpha$-duality} if for all $C \subseteq N$ and all $X \subseteq W$,
		\[
		X \in E_w(C)
		\quad\Longleftrightarrow\quad
		\overline{X} \notin E_w(\coalc).
		\]
		A coalition model $\Mod=(W,E,V)$ is \emph{$\alpha$-dual} if $E_w$ satisfies $\alpha$-duality at every state $w \in W$.
	\end{definition}
	
	At the formula level, $\alpha$-duality validates the schema
	\[
	\Eff{C}\varphi
	\leftrightarrow
	\neg\Eff{\coalc}\neg\varphi.
	\]
	Indeed,
	\[
	\Mod,w \models \Eff{C}\varphi
	\quad\Longleftrightarrow\quad
	\sem{\varphi}_{\Mod} \in E_w(C),
	\]
	and by $\alpha$-duality this is equivalent to
	\[
	\overline{\sem{\varphi}_{\Mod}} \notin E_w(\coalc).
	\]
	Since
	\[
	\overline{\sem{\varphi}_{\Mod}}
	=
	\sem{\neg\varphi}_{\Mod},
	\]
	we obtain
	\[
	\Mod,w \models \neg\Eff{\coalc}\neg\varphi.
	\]
	Conversely, the validity of this schema for all formulas entails the set-theoretic condition only under a definability assumption, namely when every subset of $W$ is definable by some formula.
	
	The right-to-left direction of $\alpha$-duality expresses complement-determination: whenever $\coalc$ cannot force $\overline{X}$, coalition $C$ can force $X$. This property is characteristic of determined two-player zero-sum settings, but it fails in general playable models. In Matching Pennies, for example, each individual player is unable to force a match and also unable to force a mismatch.
	
	We treat $\alpha$-duality as an optional strengthening of playability throughout the paper. The structure of Full Inability is most distinctive when this duality fails, because the failure of $\alpha$-duality reveals the gap between simple non-ability and complete strategic neutralization.

\section{The Four-Fold Spectrum of Strategic Power}
\label{sec:classification}

The traditional binary view of coalitional power---having or lacking an ability---is too coarse for fine-grained analysis of multi-agent interaction. This section introduces a four-fold classification obtained by simultaneously evaluating a coalition's ability to enforce a formula and its negation.

\subsection{Defining the Four Categories}
\label{subsec:four-categories}

For a fixed coalition $C$ and formula $\varphi$, two basic questions determine coalitional status:
\begin{enumerate}[label=(\arabic*)]
	\item Can $C$ enforce $\varphi$?
	\item Can $C$ enforce $\neg\varphi$?
\end{enumerate}
The two answers generate an exhaustive and mutually exclusive classification.

\begin{definition}[The Power Spectrum]
	\label{def:four-fold}
	Let $C \subseteq N$ and $\varphi \in \Lang_{\CL}$. We define:
	\begin{enumerate}[label=(\arabic*)]
		\item \textbf{Full Control}:
		\[
		\FC_C(\varphi) \;\defeq\; \Eff{C}\varphi \;\wedge\; \Eff{C}\neg\varphi.
		\]
		Coalition $C$ has two-sided control: it can determine the truth value of $\varphi$ in either direction.
		
		\item \textbf{Positive Determination}:
		\[
		\PD_C(\varphi) \;\defeq\; \Eff{C}\varphi \;\wedge\; \neg\Eff{C}\neg\varphi.
		\]
		Coalition $C$ can guarantee $\varphi$ but cannot guarantee its falsity.
		
		\item \textbf{Adverse Determination}:
		\[
		\AD_C(\varphi) \;\defeq\; \neg\Eff{C}\varphi \;\wedge\; \Eff{C}\neg\varphi.
		\]
		Coalition $C$ can guarantee $\neg\varphi$ but cannot guarantee $\varphi$.
		
		\item \textbf{Full Inability}:
		\[
		\FI_C(\varphi) \;\defeq\; \neg\Eff{C}\varphi \;\wedge\; \neg\Eff{C}\neg\varphi.
		\]
		Coalition $C$ lacks deterministic control over $\varphi$: it can enforce neither $\varphi$ nor $\neg\varphi$.
	\end{enumerate}
\end{definition}

The four categories correspond to the following truth table, where 1 denotes satisfaction and 0 denotes non-satisfaction:
\[
\begin{array}{c|c|c}
	\Eff{C}\varphi & \Eff{C}\neg\varphi & \text{Category} \\
	\hline
	1 & 1 & \FC_C(\varphi) \\
	1 & 0 & \PD_C(\varphi) \\
	0 & 1 & \AD_C(\varphi) \\
	0 & 0 & \FI_C(\varphi)
\end{array}
\]

\begin{theorem}[Exhaustiveness and Mutual Exclusivity]
	\label{thm:partition}
	For any coalition model $\Mod$, state $w$, coalition $C$, and formula $\varphi$, exactly one of
	\[
	\FC_C(\varphi),\quad \PD_C(\varphi),\quad \AD_C(\varphi),\quad \FI_C(\varphi)
	\]
	holds at $w$.
\end{theorem}

\begin{proof}
	Let
	\[
	\alpha = (\Mod,w \models \Eff{C}\varphi)
	\quad\text{and}\quad
	\beta = (\Mod,w \models \Eff{C}\neg\varphi).
	\]
	By classical bivalence, each of $\alpha$ and $\beta$ is either true or false. Hence exactly one of the four Boolean assignments
	\[
	(T,T),\quad (T,F),\quad (F,T),\quad (F,F)
	\]
	holds. By Definition~\ref{def:four-fold}, these four assignments correspond respectively to
	\[
	\FC_C(\varphi),\quad \PD_C(\varphi),\quad \AD_C(\varphi),\quad \FI_C(\varphi).
	\]
	Therefore exactly one of the four categories holds at $w$.
\end{proof}

\begin{proposition}[Negation Symmetry]
	\label{prop:negation-symmetry}
	For every coalition $C$ and formula $\varphi$, the following equivalences are valid:
	\begin{align*}
		\FC_C(\varphi) &\;\leftrightarrow\; \FC_C(\neg\varphi), \\
		\FI_C(\varphi) &\;\leftrightarrow\; \FI_C(\neg\varphi), \\
		\PD_C(\varphi) &\;\leftrightarrow\; \AD_C(\neg\varphi), \\
		\AD_C(\varphi) &\;\leftrightarrow\; \PD_C(\neg\varphi).
	\end{align*}
\end{proposition}

\begin{proof}
	The equivalences follow by expanding the definitions and using classical truth-set semantics. For every model $\Mod$,
	\[
	\sem{\neg\neg\varphi}_{\Mod} = \sem{\varphi}_{\Mod}.
	\]
	Hence, for every coalition $C$ and state $w$,
	\[
	\Mod,w \models \Eff{C}\neg\neg\varphi
	\quad\Longleftrightarrow\quad
	\Mod,w \models \Eff{C}\varphi.
	\]
	For example,
	\begin{align*}
		\FC_C(\neg\varphi)
		&\equiv \Eff{C}\neg\varphi \wedge \Eff{C}\neg\neg\varphi \\
		&\equiv \Eff{C}\neg\varphi \wedge \Eff{C}\varphi \\
		&\equiv \FC_C(\varphi).
	\end{align*}
	Similarly,
	\begin{align*}
		\FI_C(\neg\varphi)
		&\equiv \neg\Eff{C}\neg\varphi \wedge \neg\Eff{C}\neg\neg\varphi \\
		&\equiv \neg\Eff{C}\neg\varphi \wedge \neg\Eff{C}\varphi \\
		&\equiv \FI_C(\varphi),
	\end{align*}
	and the two mixed cases give
	\[
	\PD_C(\varphi) \leftrightarrow \AD_C(\neg\varphi)
	\quad\text{and}\quad
	\AD_C(\varphi) \leftrightarrow \PD_C(\neg\varphi).
	\]
\end{proof}

\begin{proposition}[Full-Control--Full-Inability Polarity]
	\label{prop:fc-fi-polarity}
	In every playable coalition model, the following schema is valid:
	\[
	\FC_C(\varphi) \to \FI_{\coalc}(\varphi).
	\]
\end{proposition}

\begin{proof}
	Assume $\Mod,w \models \FC_C(\varphi)$. Then
	\[
	\Mod,w \models \Eff{C}\varphi
	\quad\text{and}\quad
	\Mod,w \models \Eff{C}\neg\varphi.
	\]
	Equivalently,
	\[
	\sem{\varphi}_{\Mod} \in E_w(C)
	\quad\text{and}\quad
	\sem{\neg\varphi}_{\Mod} \in E_w(C).
	\]
	By Regularity (Lemma~\ref{lem:regularity}), the first inclusion yields
	\[
	\overline{\sem{\varphi}_{\Mod}} \notin E_w(\coalc).
	\]
	Since $\overline{\sem{\varphi}_{\Mod}}=\sem{\neg\varphi}_{\Mod}$, this means
	\[
	\Mod,w \not\models \Eff{\coalc}\neg\varphi.
	\]
	Similarly, from $\sem{\neg\varphi}_{\Mod} \in E_w(C)$ and Regularity we obtain
	\[
	\overline{\sem{\neg\varphi}_{\Mod}} \notin E_w(\coalc).
	\]
	Since $\overline{\sem{\neg\varphi}_{\Mod}}=\sem{\varphi}_{\Mod}$, this gives
	\[
	\Mod,w \not\models \Eff{\coalc}\varphi.
	\]
	Therefore
	\[
	\Mod,w \models
	\neg\Eff{\coalc}\varphi
	\wedge
	\neg\Eff{\coalc}\neg\varphi,
	\]
	that is,
	\[
	\Mod,w \models \FI_{\coalc}(\varphi).
	\]
\end{proof}

\begin{remark}
	Proposition~\ref{prop:fc-fi-polarity} establishes only one-way polarity. In arbitrary playable models, $\FI_{\coalc}(\varphi)$ does not entail $\FC_C(\varphi)$. The converse requires the stronger assumption of $\alpha$-duality. This asymmetry is one of the main reasons Full Inability cannot be reduced to the simple negation of ability.
\end{remark}

\subsection{Canonical Examples}
\label{subsec:case-studies}

The following examples illustrate the four categories via local strategic game forms.

\begin{example}[Dictatorship]
	\label{ex:dictatorship}
	At a state $w$, suppose player $1$ has two actions, one guaranteeing $p$ and one guaranteeing $\neg p$, while player $2$'s actions do not affect whether $p$ holds. Then
	\[
	\FC_{\{1\}}(p)
	\quad\text{and}\quad
	\FI_{\{2\}}(p)
	\]
	hold at $w$: player $1$ has two-sided control over $p$, while player $2$ has no standalone deterministic control over it.
\end{example}

\begin{example}[Matching Pennies]
	\label{ex:matching-pennies}
	Consider the standard Matching Pennies game, and let $p$ express that the two coins match. Each individual player $i \in \{1,2\}$ satisfies
	\[
	\FI_{\{i\}}(p):
	\]
	by choosing heads or tails alone, player $i$ cannot guarantee either a match or a mismatch, since the other player's action may change the outcome. By contrast, the grand coalition satisfies
	\[
	\FC_N(p),
	\]
	because the joint action $(H,H)$ guarantees a match, whereas the joint action $(H,T)$ guarantees a mismatch. Thus Full Inability at the individual level can coexist with Full Control at the grand-coalition level.
\end{example}

\begin{example}[Veto Power]
	\label{ex:veto}
	In a unanimous-consent board, let $p$ mean that a motion passes. Any individual member $i$ can block the motion, thereby enforcing $\neg p$, but cannot alone ensure that the motion passes, since all other members must consent. Hence
	\[
	\AD_{\{i\}}(p)
	\]
	holds.
\end{example}

\begin{example}[Positive Determination]
	\label{ex:pd}
	Let $p$ mean that an emergency shutdown occurs. Suppose monitor $i$ can trigger the shutdown, but other monitors can also independently trigger it. Then $i$ can guarantee $p$ by triggering the shutdown, but cannot guarantee $\neg p$, since another monitor may still trigger it. Therefore
	\[
	\PD_{\{i\}}(p)
	\]
	holds.
\end{example}

\subsection{Connection to Social Choice: Dummy Players}
\label{subsec:dummy-players}

In cooperative game theory and social choice, a \emph{dummy player} is one whose presence does not change the relevant coalitional outcome, such as a coalition's winning status \cite{BCG16, Kurz22}. We formulate a local, proposition-relative analogue within Coalition Logic.

\begin{definition}[Propositional Dummy Player]
	\label{def:dummy}
	Let $\Mod = (W,E,V)$ be a coalition model and $w \in W$. Agent $i \in N$ is a \emph{dummy player with respect to $\varphi$} at $w$ if for every $C \subseteq N \setminus \{i\}$,
	\[
	\Mod,w \models \Eff{C \cup \{i\}}\varphi
	\;\Longleftrightarrow\;
	\Mod,w \models \Eff{C}\varphi.
	\]
\end{definition}

\begin{proposition}[Dummyhood--Full-Inability Connection]
	\label{prop:dummy-fi}
	Let $\Mod = (W,E,V)$ be a coalition model, $w \in W$, and $i \in N$.
	\begin{enumerate}[label=(\arabic*)]
		\item If $i$ is a dummy player with respect to $\varphi$ at $w$ and
		\[
		\Mod,w \not\models \Eff{\emptyset}\varphi,
		\]
		then
		\[
		\Mod,w \not\models \Eff{\{i\}}\varphi.
		\]
		
		\item If $i$ is a dummy player with respect to both $\varphi$ and $\neg\varphi$ at $w$, and
		\[
		\Mod,w \not\models \Eff{\emptyset}\varphi
		\quad\text{and}\quad
		\Mod,w \not\models \Eff{\emptyset}\neg\varphi,
		\]
		then
		\[
		\Mod,w \models \FI_{\{i\}}(\varphi).
		\]
	\end{enumerate}
\end{proposition}

\begin{proof}
	For (1), instantiate Definition~\ref{def:dummy} with $C=\emptyset$. Since $\emptyset \subseteq N\setminus\{i\}$, we have
	\[
	\Mod,w \models \Eff{\{i\}}\varphi
	\;\Longleftrightarrow\;
	\Mod,w \models \Eff{\emptyset}\varphi.
	\]
	The premise
	\[
	\Mod,w \not\models \Eff{\emptyset}\varphi
	\]
	therefore yields
	\[
	\Mod,w \not\models \Eff{\{i\}}\varphi.
	\]
	
	For (2), apply (1) first to $\varphi$ and then to $\neg\varphi$. We obtain
	\[
	\Mod,w \not\models \Eff{\{i\}}\varphi
	\quad\text{and}\quad
	\Mod,w \not\models \Eff{\{i\}}\neg\varphi.
	\]
	By Definition~\ref{def:four-fold}, this is exactly
	\[
	\Mod,w \models \FI_{\{i\}}(\varphi).
	\]
\end{proof}

\begin{remark}
	Full Inability is a necessary consequence of two-sided propositional dummyhood when the empty coalition cannot determine either side of the proposition. The converse does not hold: $\FI_{\{i\}}(\varphi)$ does not imply dummyhood. An agent may lack individual control over $\varphi$ while remaining pivotal in larger coalitions. Matching Pennies illustrates this: each player satisfies $\FI_{\{i\}}(p)$ individually, but each is strategically relevant when coordinating with the other. Full Inability therefore captures the absence of standalone deterministic control, whereas dummyhood is a stronger coalitional invariance property.
\end{remark}

	\section{Polarity, Symmetry, and Conditional Duality}
	\label{sec:duality-theory}
	
	This section examines structural relations among the four categories of the power spectrum. The central point is that the relation between Full Control and Full Inability is not a genuine duality in arbitrary playable models. Playability yields only a one-way polarity: if a coalition has Full Control, then its complement has Full Inability. A full dual equivalence emerges only under the stronger assumption of $\alpha$-duality.
	
	\subsection{One-Way Polarity: Control Entails Complementary Inability}
	\label{subsec:one-way-duality}
	
	In any playable coalition model, Full Control for a coalition entails Full Inability for its complementary coalition. This follows directly from regularity.
	
	\begin{theorem}[One-Way Polarity]
		\label{thm:one-way-duality}
		For any playable coalition model $\Mod$, state $w$, coalition $C \subseteq N$, and formula $\varphi$,
		\[
		\Mod,w \models \FC_C(\varphi)
		\quad\Longrightarrow\quad
		\Mod,w \models \FI_{\coalc}(\varphi).
		\]
	\end{theorem}
	
	\begin{proof}
		This is Proposition~\ref{prop:fc-fi-polarity}, restated here to emphasize its role in the duality analysis.
	\end{proof}
	
	\subsection{The Matching Pennies Counterexample}
	\label{subsec:mp-counterexample}
	
	The converse of Theorem~\ref{thm:one-way-duality} fails in arbitrary playable models.
	
	\begin{observation}[Mutual Inability in Matching Pennies]
		\label{obs:matching-pennies-counterexample}
		Consider a two-player Matching Pennies game, and let $p$ express that the two coins match. Each individual player $i \in \{1,2\}$ satisfies
		\[
		\FI_{\{i\}}(p),
		\]
		since neither player can guarantee a match or a mismatch alone. However, neither individual player has Full Control:
		\[
		\Mod,w \not\models \FC_{\{1\}}(p)
		\quad\text{and}\quad
		\Mod,w \not\models \FC_{\{2\}}(p).
		\]
		Taking $C=\{2\}$ and hence $\coalc=\{1\}$, we have
		\[
		\Mod,w \models \FI_{\coalc}(p)
		\quad\text{but}\quad
		\Mod,w \not\models \FC_C(p).
		\]
		Thus the converse implication
		\[
		\FI_{\coalc}(p) \to \FC_C(p)
		\]
		is not valid in all playable models.
	\end{observation}
	
	This example also clarifies the difference between Full Inability and the mere negation of Full Control:
	\[
	\FI_C(\varphi)
	=
	\neg\Eff{C}\varphi
	\land
	\neg\Eff{C}\neg\varphi,
	\]
	whereas
	\[
	\neg\FC_C(\varphi)
	=
	\neg\Eff{C}\varphi
	\lor
	\neg\Eff{C}\neg\varphi.
	\]
	Hence
	\[
	\FI_C(\varphi) \to \neg\FC_C(\varphi)
	\]
	is valid, but the converse is not. Full Inability captures complete non-determination with respect to the issue $\varphi$: coalition $C$ can enforce neither side of the issue.
	
	\subsection{Dual Equivalence under $\alpha$-Duality}
	\label{subsec:conditional-duality}
	
	Full equivalence between Full Inability and complementary Full Control holds under $\alpha$-duality.
	
	\begin{theorem}[Conditional Dual Equivalence]
		\label{thm:conditional-duality}
		Let $\Mod=(W,E,V)$ be a coalition model such that $E_w$ satisfies $\alpha$-duality for every $w \in W$. Then for every state $w$, coalition $C \subseteq N$, and formula $\varphi$,
		\[
		\Mod,w \models \FI_C(\varphi)
		\;\Longleftrightarrow\;
		\Mod,w \models \FC_{\coalc}(\varphi).
		\]
	\end{theorem}
	
	\begin{proof}
		By Definition~\ref{def:four-fold},
		\[
		\FI_C(\varphi)
		\equiv
		\neg\Eff{C}\varphi
		\land
		\neg\Eff{C}\neg\varphi.
		\]
		By the formula-level consequence of $\alpha$-duality,
		\[
		\Eff{C}\varphi
		\leftrightarrow
		\neg\Eff{\coalc}\neg\varphi.
		\]
		Therefore,
		\[
		\neg\Eff{C}\varphi
		\leftrightarrow
		\Eff{\coalc}\neg\varphi.
		\]
		Applying the same principle to $\neg\varphi$ gives
		\[
		\Eff{C}\neg\varphi
		\leftrightarrow
		\neg\Eff{\coalc}\neg\neg\varphi.
		\]
		Since $\neg\neg\varphi$ is semantically equivalent to $\varphi$, this yields
		\[
		\Eff{C}\neg\varphi
		\leftrightarrow
		\neg\Eff{\coalc}\varphi,
		\]
		and hence
		\[
		\neg\Eff{C}\neg\varphi
		\leftrightarrow
		\Eff{\coalc}\varphi.
		\]
		Substituting both equivalences into the definition of $\FI_C(\varphi)$, we obtain
		\[
		\FI_C(\varphi)
		\leftrightarrow
		\Eff{\coalc}\neg\varphi
		\land
		\Eff{\coalc}\varphi.
		\]
		By commutativity of conjunction, the right-hand side is exactly
		\[
		\FC_{\coalc}(\varphi).
		\]
		Thus
		\[
		\Mod,w \models \FI_C(\varphi)
		\;\Longleftrightarrow\;
		\Mod,w \models \FC_{\coalc}(\varphi).
		\]
	\end{proof}
	
	\begin{corollary}[Complementary Full Control and Full Inability]
		\label{cor:fc-fi-alpha}
		Under $\alpha$-duality, for every coalition $C$ and formula $\varphi$,
		\[
		\FC_C(\varphi)
		\leftrightarrow
		\FI_{\coalc}(\varphi).
		\]
	\end{corollary}
	
	\begin{proof}
		Apply Theorem~\ref{thm:conditional-duality} to the complementary coalition $\coalc$. Since $\overline{\coalc}=C$, we obtain
		\[
		\FI_{\coalc}(\varphi)
		\leftrightarrow
		\FC_C(\varphi).
		\]
	\end{proof}
	
	\subsection{A Conditional Klein Four-Group Symmetry}
	\label{subsec:klein-group}
	
	Two natural transformations act on coalition--formula pairs:
	\[
	f_{\mathsf{neg}}(C,\varphi) = (C,\neg\varphi)
	\]
	and
	\[
	f_{\mathsf{comp}}(C,\varphi) = (\coalc,\varphi).
	\]
	Their composition is
	\[
	f_{\mathsf{both}}
	=
	f_{\mathsf{neg}}\circ f_{\mathsf{comp}}.
	\]
	Strictly speaking, $f_{\mathsf{neg}}^2(C,\varphi)=(C,\neg\neg\varphi)$ is identical to $(C,\varphi)$ only up to semantic equivalence. Thus the following symmetry should be understood as an action on category labels, or on formulas modulo classical semantic equivalence.
	
	Let
	\[
	a = (\Mod,w \models \Eff{C}\varphi)
	\quad\text{and}\quad
	b = (\Mod,w \models \Eff{C}\neg\varphi).
	\]
	The four categories correspond to the Boolean pairs:
	\[
	\FC=(T,T),\qquad
	\PD=(T,F),\qquad
	\AD=(F,T),\qquad
	\FI=(F,F).
	\]
	
	\begin{theorem}[Conditional Klein Four-Group Symmetry]
		\label{thm:klein-symmetry}
		In $\alpha$-dual coalition models, the transformations
		\[
		\mathrm{id},\quad
		f_{\mathsf{neg}},\quad
		f_{\mathsf{comp}},\quad
		f_{\mathsf{both}}
		\]
		induce permutations of the four power-spectrum categories. These permutations form a group isomorphic to the Klein four-group $V_4$.
	\end{theorem}
	
	\begin{proof}
		First consider $f_{\mathsf{neg}}$. Replacing $\varphi$ by $\neg\varphi$ transforms the Boolean coordinates as
		\[
		(a,b) \mapsto (b,a),
		\]
		because the first coordinate becomes $\Eff{C}\neg\varphi$ and the second becomes $\Eff{C}\neg\neg\varphi$, which is semantically equivalent to $\Eff{C}\varphi$. Hence $f_{\mathsf{neg}}$ fixes $\FC$ and $\FI$, and exchanges $\PD$ and $\AD$.
		
		Next consider $f_{\mathsf{comp}}$. Under $\alpha$-duality,
		\[
		\Eff{\coalc}\varphi
		\leftrightarrow
		\neg\Eff{C}\neg\varphi
		\]
		and
		\[
		\Eff{\coalc}\neg\varphi
		\leftrightarrow
		\neg\Eff{C}\varphi.
		\]
		Therefore the coordinates of $(\coalc,\varphi)$ are
		\[
		(\neg b,\neg a).
		\]
		Thus $f_{\mathsf{comp}}$ acts as
		\[
		(a,b)\mapsto(\neg b,\neg a).
		\]
		It exchanges $\FC$ and $\FI$, while fixing $\PD$ and $\AD$.
		
		Finally,
		\[
		f_{\mathsf{both}}
		=
		f_{\mathsf{neg}}\circ f_{\mathsf{comp}}
		\]
		acts as
		\[
		(a,b)\mapsto(\neg a,\neg b).
		\]
		Hence it exchanges $\FC$ with $\FI$, and exchanges $\PD$ with $\AD$.
		
		The induced action on category labels is summarized in Table~\ref{tab:klein-action}. Each non-identity transformation is involutive, and the two generators commute:
		\[
		f_{\mathsf{neg}}^2
		=
		f_{\mathsf{comp}}^2
		=
		\mathrm{id},
		\qquad
		f_{\mathsf{neg}}\circ f_{\mathsf{comp}}
		=
		f_{\mathsf{comp}}\circ f_{\mathsf{neg}}.
		\]
		Consequently, the four transformations form a group isomorphic to the Klein four-group $V_4$.
	\end{proof}
	
	\begin{table}[ht]
		\centering
		\begin{tabular}{c|cccc}
			\toprule
			Transform & $\FC$ & $\PD$ & $\AD$ & $\FI$ \\
			\midrule
			$f_{\mathsf{neg}}$ & $\FC$ & $\AD$ & $\PD$ & $\FI$ \\
			$f_{\mathsf{comp}}$ & $\FI$ & $\PD$ & $\AD$ & $\FC$ \\
			$f_{\mathsf{both}} = f_{\mathsf{neg}} \circ f_{\mathsf{comp}}$ & $\FI$ & $\AD$ & $\PD$ & $\FC$ \\
			\bottomrule
		\end{tabular}
		\caption{Action of transformations on power-spectrum categories under $\alpha$-duality.}
		\label{tab:klein-action}
	\end{table}
	
	\subsection{Failure of Category-Level Complementation without $\alpha$-Duality}
	\label{subsec:failure-complementation}
	
	In arbitrary playable models, the above group action breaks down. The transformation $f_{\mathsf{neg}}$ remains well defined at the level of category labels, since negation merely swaps the two coordinates:
	\[
	(a,b)\mapsto(b,a).
	\]
	By contrast, $f_{\mathsf{comp}}$ need not induce a permutation of category labels. Without $\alpha$-duality, the category of $(\coalc,\varphi)$ is not determined solely by the category of $(C,\varphi)$.
	
	Playability provides only the regularity constraints
	\[
	\Eff{C}\varphi
	\Rightarrow
	\neg\Eff{\coalc}\neg\varphi
	\]
	and
	\[
	\Eff{C}\neg\varphi
	\Rightarrow
	\neg\Eff{\coalc}\varphi.
	\]
	Consequently,
	\[
	\FC_C(\varphi)
	\Rightarrow
	\FI_{\coalc}(\varphi),
	\]
	but the converse need not hold.
	
	More generally, the possible complementary categories are constrained only partially:
	\begin{itemize}
		\item If $\FC_C(\varphi)$ holds, then $\FI_{\coalc}(\varphi)$ must hold.
		
		\item If $\PD_C(\varphi)$ holds, then $\coalc$ cannot enforce $\neg\varphi$, but may or may not enforce $\varphi$. Thus $\coalc$ is in $\PD$ or $\FI$.
		
		\item If $\AD_C(\varphi)$ holds, then $\coalc$ cannot enforce $\varphi$, but may or may not enforce $\neg\varphi$. Thus $\coalc$ is in $\AD$ or $\FI$.
		
		\item If $\FI_C(\varphi)$ holds, playability alone imposes no category-level constraint on $\coalc$.
	\end{itemize}
	
	This failure of coalition complementation to induce a category-level permutation is precisely the structural gap between merely playable models and complement-determined, $\alpha$-dual models.

	\section{Geometrical and Lattice-Theoretic Properties}
	\label{sec:structural-properties}
	
	This section moves from formula-level classifications to a set-theoretic characterization within the powerset lattice $(\powerset(W),\subseteq)$. By identifying each formula $\varphi$ with its truth set
	\[
	\sem{\varphi}_{\Mod}
	=
	\{v \in W \mid \Mod,v \models \varphi\},
	\]
	the four power categories appear as order-theoretic regions in the Boolean algebra $\powerset(W)$. A central result is the convexity theorem: each region is order-convex. In particular, Full Inability forms an interval-stable region---if two outcome sets lie in the Full Inability region, then every intermediate set (in the subset order) also lies in that region. We formalize this property as \emph{Strategic Contiguity}.
	
	\subsection{The Order-Theoretic Polarity of Effectivity}
	\label{subsec:effectivity-regions}
	
	Fix a playable coalition model $\Mod=(W,E,V)$, a state $w \in W$, and a coalition $C \subseteq N$. Recall that $E_w(C) \subseteq \powerset(W)$ denotes the family of outcome sets enforceable by $C$ at $w$. Define the co-effectivity region of $C$ by
	\[
	E_w^*(C)
	\defeq
	\{X \subseteq W \mid \overline{X} \in E_w(C)\},
	\]
	where $\overline{X}=W\setminus X$.
	
	Intuitively, $E_w(C)$ is the region of outcome sets that $C$ can enforce, while $E_w^*(C)$ is the region of outcome sets whose complements $C$ can enforce. Thus $X \in E_w^*(C)$ means that $C$ can enforce the failure of $X$.
	
	\begin{lemma}[Order Polarity of Effectivity Regions]
		\label{lem:lattice-polarity}
		For any playable coalition model $\Mod$, state $w$, and coalition $C \subseteq N$:
		\begin{enumerate}[label=(\arabic*)]
			\item $E_w(C)$ is upward closed in $(\powerset(W),\subseteq)$: if $X \in E_w(C)$ and $X \subseteq Y$, then $Y \in E_w(C)$.
			
			\item $E_w^*(C)$ is downward closed in $(\powerset(W),\subseteq)$: if $X \in E_w^*(C)$ and $Y \subseteq X$, then $Y \in E_w^*(C)$.
		\end{enumerate}
	\end{lemma}
	
	\begin{proof}
		Part~(1) is outcome monotonicity for playable effectivity functions.
		
		For~(2), suppose $X \in E_w^*(C)$ and $Y \subseteq X$. By definition of $E_w^*(C)$, we have $\overline{X} \in E_w(C)$. Since $Y \subseteq X$, it follows that $\overline{X} \subseteq \overline{Y}$. By upward closure of $E_w(C)$, we obtain $\overline{Y} \in E_w(C)$. Hence $Y \in E_w^*(C)$.
	\end{proof}
	
	The four power categories induce four regions in the powerset lattice:
	\begin{align*}
		\Rfc^{w}(C)
		&\defeq
		E_w(C) \cap E_w^*(C), \\
		\Rpd^{w}(C)
		&\defeq
		E_w(C) \setminus E_w^*(C), \\
		\Rad^{w}(C)
		&\defeq
		E_w^*(C) \setminus E_w(C), \\
		\Rfi^{w}(C)
		&\defeq
		\powerset(W) \setminus \bigl(E_w(C) \cup E_w^*(C)\bigr).
	\end{align*}
	Consequently, for any formula $\varphi$,
	\[
	\Mod,w \models \FC_C(\varphi)
	\quad\Longleftrightarrow\quad
	\sem{\varphi}_{\Mod} \in \Rfc^{w}(C),
	\]
	and analogously for $\PD_C(\varphi)$, $\AD_C(\varphi)$, and $\FI_C(\varphi)$.
	
	Figure~\ref{fig:lattice-power-regions} gives a schematic Venn-style visualization of this decomposition.
	
	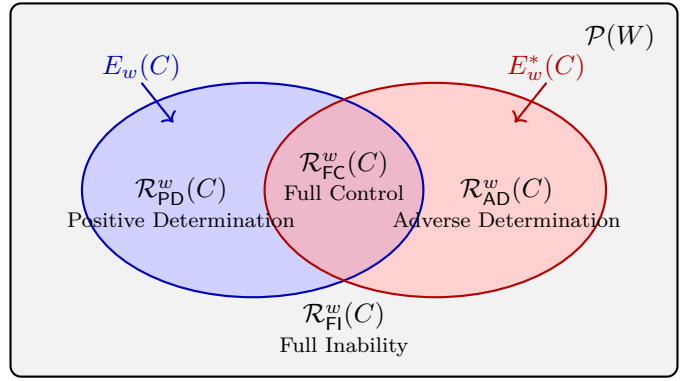
\begin{figure}[ht]
		\centering
		\begin{tikzpicture}[scale=1.05, every node/.style={font=\small}]
			\fill[gray!10, rounded corners] (-4.2,-2.35) rectangle (4.2,2.35);
			\draw[thick, rounded corners] (-4.2,-2.35) rectangle (4.2,2.35);
			\node[anchor=north east] at (4.05,2.22) {$\powerset(W)$};
			
			\fill[blue!18] (-1.15,0) ellipse (2.15 and 1.35);
			\fill[red!18] (1.15,0) ellipse (2.15 and 1.35);
			
			\begin{scope}
				\clip (-1.15,0) ellipse (2.15 and 1.35);
				\fill[purple!28] (1.15,0) ellipse (2.15 and 1.35);
			\end{scope}
			
			\draw[blue!70!black, thick] (-1.15,0) ellipse (2.15 and 1.35);
			\draw[red!70!black, thick] (1.15,0) ellipse (2.15 and 1.35);
			
			\node[blue!70!black] at (-2.55,1.55) {$E_w(C)$};
			\node[red!70!black] at (2.55,1.55) {$E_w^*(C)$};
			
			\node[align=center] at (0,0.18)
			{$\Rfc^{w}(C)$\\[-1pt]\scriptsize Full Control};
			
			\node[align=center] at (-2.05,-0.15)
			{$\Rpd^{w}(C)$\\[-1pt]\scriptsize Positive Determination};
			
			\node[align=center] at (2.05,-0.15)
			{$\Rad^{w}(C)$\\[-1pt]\scriptsize Adverse Determination};
			
			\node[align=center] at (0,-1.75)
			{$\Rfi^{w}(C)$\\[-1pt]\scriptsize Full Inability};
			
			\draw[->, blue!70!black, thick] (-2.55,1.35) -- (-2.15,0.85);
			\draw[->, red!70!black, thick] (2.55,1.35) -- (2.15,0.85);
		\end{tikzpicture}
		\caption{Schematic Venn-style decomposition of $\powerset(W)$ into four strategic regions. Note: $E_w(C)$ is upward closed, whereas $E_w^*(C)$ is downward closed.}
		\label{fig:lattice-power-regions}
	\end{figure}

	\subsection{A Strategic Bilattice Representation}
	\label{subsec:strategic-bilattice}
	
	The four-fold spectrum admits a natural algebraic representation as a Belnap--Dunn style bilattice. This formalization makes explicit that two basic monotonicities in Coalition Logic operate independently: outcome inclusion moves strategic values along a \emph{directionality order}, whereas coalition expansion moves them along a \emph{determination order}.
	
	Fix a playable coalition model $\Mod=(W,E,V)$, a state $w \in W$, and a coalition $C \subseteq N$. For every outcome set $X \subseteq W$, define the \emph{strategic value} of $X$ for $C$ at $w$ by
	\[
	\nu_C^w(X)
	\defeq
	(a_C^w(X),b_C^w(X))
	\in \{0,1\}^2,
	\]
	where
	\[
	a_C^w(X)=1
	\quad\Longleftrightarrow\quad
	X \in E_w(C),
	\]
	and
	\[
	b_C^w(X)=1
	\quad\Longleftrightarrow\quad
	\overline{X} \in E_w(C).
	\]
	The four strategic categories correspond to:
	\[
	\FC=(1,1),
	\qquad
	\PD=(1,0),
	\qquad
	\AD=(0,1),
	\qquad
	\FI=(0,0).
	\]
	
	\begin{definition}[Strategic Bilattice Orders]
		\label{def:strategic-bilattice-orders}
		Let
		\[
		\mathbb{B}_{\mathsf{str}}
		=
		\{\FI,\PD,\AD,\FC\}
		\cong
		\{0,1\}^2.
		\]
		We define two partial orders:
		\begin{enumerate}[label=(\roman*)]
			\item The \emph{determination order} $\leq_k$:
			\[
			(a,b) \leq_k (a',b')
			\quad\text{iff}\quad
			a \leq a'
			\text{ and }
			b \leq b'.
			\]
			
			\item The \emph{directionality order} $\leq_t$:
			\[
			(a,b) \leq_t (a',b')
			\quad\text{iff}\quad
			a \leq a'
			\text{ and }
			b' \leq b.
			\]
		\end{enumerate}
	\end{definition}
	
	Under the determination order,
	\[
	\FI \leq_k \PD,\AD \leq_k \FC.
	\]
	Thus $\FI$ is least determined, $\FC$ is maximally determined, and $\PD,\AD$ are one-sided. Under the directionality order,
	\[
	\AD \leq_t \FI,\FC \leq_t \PD.
	\]
	Here $\AD$ is maximally falsity-directed, $\PD$ is maximally truth-directed, while $\FI$ and $\FC$ are directionally neutral.
	
	As an ordered bilattice, $(\mathbb{B}_{\mathsf{str}},\leq_k,\leq_t)$ is isomorphic to the standard Belnap--Dunn bilattice $\mathcal{FOUR}$. Under this isomorphism, $\PD$ corresponds to truth, $\AD$ to falsity, $\FC$ to overdetermination, and $\FI$ to underdetermination. This is an order-theoretic correspondence; no additional Belnap--Dunn semantic interpretation is assumed.
	
	\begin{theorem}[Bimonotonicity of Strategic Values]
		\label{thm:bimonotonicity}
		Let $\Mod$ be a playable coalition model, $w \in W$, $C,D \subseteq N$, and $X,Y \subseteq W$.
		\begin{enumerate}[label=(\arabic*)]
			\item If $X \subseteq Y$, then
			\[
			\nu_C^w(X) \leq_t \nu_C^w(Y).
			\]
			
			\item If $C \subseteq D$, then
			\[
			\nu_C^w(X) \leq_k \nu_D^w(X).
			\]
		\end{enumerate}
	\end{theorem}
	
	\begin{proof}
		For~(1), suppose $X \subseteq Y$. If $a_C^w(X)=1$, then $X \in E_w(C)$. Since $E_w(C)$ is upward closed, $Y \in E_w(C)$, so $a_C^w(Y)=1$. Hence
		\[
		a_C^w(X) \leq a_C^w(Y).
		\]
		For the second coordinate, suppose $b_C^w(Y)=1$. Then $\overline{Y} \in E_w(C)$. Since $X \subseteq Y$, we have $\overline{Y} \subseteq \overline{X}$. By upward closure, $\overline{X} \in E_w(C)$, hence $b_C^w(X)=1$. Thus
		\[
		b_C^w(Y) \leq b_C^w(X).
		\]
		Therefore $\nu_C^w(X) \leq_t \nu_C^w(Y)$.
		
		For~(2), by Coalition Monotonicity (Lemma~\ref{lem:coalition-monotonicity}), $E_w(C) \subseteq E_w(D)$. Hence if $X \in E_w(C)$, then $X \in E_w(D)$, and if $\overline{X} \in E_w(C)$, then $\overline{X} \in E_w(D)$. Both coordinates are non-decreasing, so
		\[
		\nu_C^w(X) \leq_k \nu_D^w(X).
		\]
	\end{proof}
	
	The first part of Theorem~\ref{thm:bimonotonicity} constrains how strategic values may change under outcome-set inclusion.
	
	\begin{table}[ht]
		\centering
		\begin{tabular}{c|c}
			\toprule
			Initial value of $X$ & Possible values of $Y$ when $X \subseteq Y$ \\
			\midrule
			$\AD$ & $\AD, \FI, \FC, \PD$ \\
			$\FI$ & $\FI, \PD$ \\
			$\FC$ & $\FC, \PD$ \\
			$\PD$ & $\PD$ \\
			\bottomrule
		\end{tabular}
		\caption{Strategic value migration under outcome inclusion.}
		\label{tab:outcome-migration}
	\end{table}
	
	The second part constrains how strategic values may change under coalition expansion.
	
	\begin{table}[ht]
		\centering
		\begin{tabular}{c|c}
			\toprule
			Initial value for $C$ & Possible values for $D$ when $C \subseteq D$ \\
			\midrule
			$\FI$ & $\FI, \PD, \AD, \FC$ \\
			$\PD$ & $\PD, \FC$ \\
			$\AD$ & $\AD, \FC$ \\
			$\FC$ & $\FC$ \\
			\bottomrule
		\end{tabular}
		\caption{Strategic value migration under coalition inclusion.}
		\label{tab:coalition-migration}
	\end{table}
	
	This bilattice perspective gives an algebraic explanation for convexity: each geometric region is a fiber of the strategic valuation map.
	
	\begin{corollary}[Fiber Convexity]
		\label{cor:fiber-convexity}
		For every playable coalition model $\Mod$, state $w$, coalition $C \subseteq N$, and strategic value $s \in \{\FC,\PD,\AD,\FI\}$, the fiber
		\[
		(\nu_C^w)^{-1}(s)
		=
		\{X \subseteq W \mid \nu_C^w(X)=s\}
		\]
		is order-convex in $(\powerset(W),\subseteq)$.
	\end{corollary}
	
	\begin{proof}
		Suppose $X \subseteq Y \subseteq Z$ and
		\[
		\nu_C^w(X)=s=\nu_C^w(Z).
		\]
		By Theorem~\ref{thm:bimonotonicity}(1),
		\[
		\nu_C^w(X) \leq_t \nu_C^w(Y) \leq_t \nu_C^w(Z).
		\]
		Therefore
		\[
		s \leq_t \nu_C^w(Y) \leq_t s.
		\]
		By antisymmetry of $\leq_t$, we obtain
		\[
		\nu_C^w(Y)=s.
		\]
		Hence the fiber is order-convex.
	\end{proof}
	
	Since
	\[
	(\nu_C^w)^{-1}(\FC)=\Rfc^w(C),
	\]
	and similarly for $\PD,\AD,\FI$, Corollary~\ref{cor:fiber-convexity} gives an algebraic proof of order-convexity. We re-establish this result via explicit set-theoretic arguments in Theorem~\ref{thm:convexity}.

	\subsection{Strategic Box and the Non-Kripkean Nature of Full Inability}
	\label{subsubsec:strategic-box}
	
	The four-fold spectrum contrasts with classical $\Box/\Diamond$ duality. Define the strategic dual by
	\[
	\Box_C\varphi
	\defeq
	\neg\Eff{C}\neg\varphi.
	\]
	This operator states that $C$ cannot force the failure of $\varphi$. The four categories can be rewritten as:
	\[
	\begin{array}{c|cc}
		& \Box_C\varphi & \neg\Box_C\varphi \\
		\hline
		\Eff{C}\varphi & \PD_C(\varphi) & \FC_C(\varphi) \\[2pt]
		\neg\Eff{C}\varphi & \FI_C(\varphi) & \AD_C(\varphi)
	\end{array}
	\]
	
	Full Inability captures a form of strategic neutralization: $C$ cannot enforce $\varphi$, but neither can it enforce the failure of $\varphi$. This is genuinely non-Kripkean. On serial Kripke frames,
	\[
	\neg\Diamond_C\varphi
	\land
	\neg\Diamond_C\neg\varphi
	\]
	is unsatisfiable. Indeed, seriality guarantees at least one accessible successor, and that successor satisfies either $\varphi$ or $\neg\varphi$. Hence at least one of $\Diamond_C\varphi$ or $\Diamond_C\neg\varphi$ must hold. In Coalition Logic, by contrast, $\FI_C(\varphi)$ is satisfiable: in Matching Pennies, an individual player lacks the capacity to force either a match or a mismatch.
	
	\subsection{Strategy-Cell Boundary Interpretation}
	\label{subsec:strategy-cell-boundary}
	
	The four-fold spectrum admits a geometric interpretation in strategic game forms. Assume $E_w$ is induced by a game form at state $w$. Let
	\[
	Act_C
	\defeq
	\prod_{i \in C} Act_i
	\]
	denote the joint actions available to $C$. For each strategy $s_C \in Act_C$, define its outcome cell by
	\[
	O_w(s_C)
	\defeq
	\{o_w(s_C,s_{\coalc}) \mid s_{\coalc} \in Act_{\coalc}\}.
	\]
	Let
	\[
	\mathcal{O}_C^w
	\defeq
	\{O_w(s_C) \mid s_C \in Act_C\}.
	\]
	Under game-form semantics,
	\[
	X \in E_w(C)
	\quad\Longleftrightarrow\quad
	\exists O \in \mathcal{O}_C^w:\; O \subseteq X.
	\]
	
	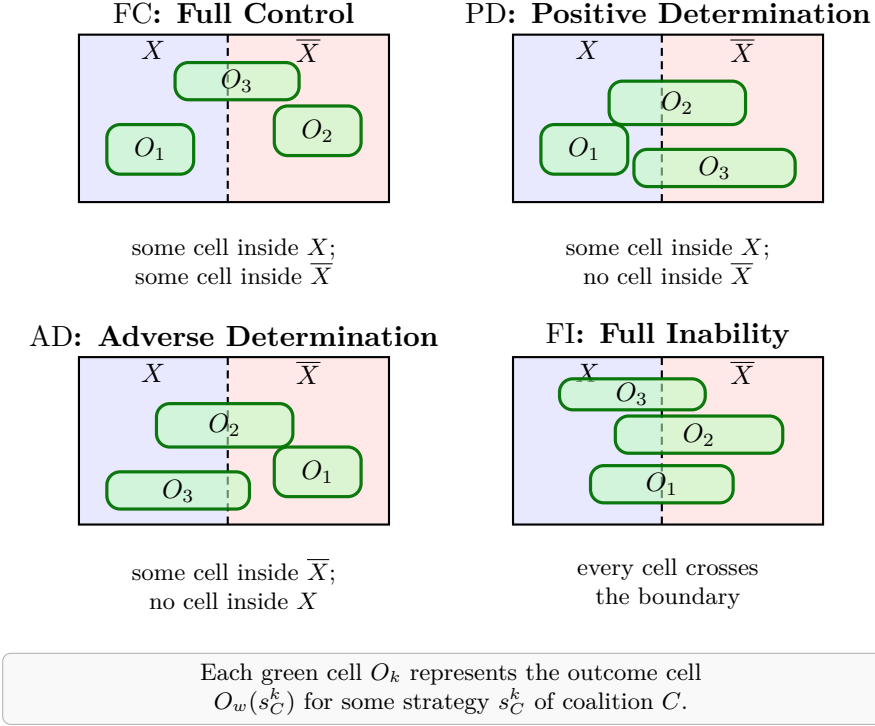
\begin{figure}[H]
		\centering
		\begin{tikzpicture}[
			scale=0.82,
			every node/.style={font=\small},
			cell/.style={draw=green!45!black, fill=green!25, rounded corners=5pt, very thick, fill opacity=.55, draw opacity=.95},
			boundary/.style={densely dashed, thick},
			panelnote/.style={align=center, font=\footnotesize, text width=4.6cm}
			]
			
			\begin{scope}[shift={(0,5.2)}]
				\node[font=\bfseries] at (2.5,3.05) {$\mathrm{FC}$: Full Control};
				\fill[blue!9] (0,0) rectangle (2.4,2.7);
				\fill[red!9] (2.4,0) rectangle (5,2.7);
				\draw[thick] (0,0) rectangle (5,2.7);
				\draw[boundary] (2.4,0) -- (2.4,2.7);
				\node at (1.2,2.45) {$X$};
				\node at (3.7,2.45) {$\overline X$};
				\draw[cell] (0.45,0.45) rectangle (1.85,1.25);
				\node at (1.15,0.85) {$O_1$};
				\draw[cell] (3.15,0.75) rectangle (4.55,1.55);
				\node at (3.85,1.15) {$O_2$};
				\draw[cell] (1.55,1.65) rectangle (3.55,2.25);
				\node at (2.55,1.95) {$O_3$};
				\node[panelnote] at (2.5,-0.95) {some cell inside $X$;\\ some cell inside $\overline X$};
			\end{scope}
			
			\begin{scope}[shift={(7.0,5.2)}]
				\node[font=\bfseries] at (2.5,3.05) {$\mathrm{PD}$: Positive Determination};
				\fill[blue!9] (0,0) rectangle (2.4,2.7);
				\fill[red!9] (2.4,0) rectangle (5,2.7);
				\draw[thick] (0,0) rectangle (5,2.7);
				\draw[boundary] (2.4,0) -- (2.4,2.7);
				\node at (1.2,2.45) {$X$};
				\node at (3.7,2.45) {$\overline X$};
				\draw[cell] (0.45,0.45) rectangle (1.85,1.25);
				\node at (1.15,0.85) {$O_1$};
				\draw[cell] (1.55,1.25) rectangle (3.75,1.95);
				\node at (2.65,1.60) {$O_2$};
				\draw[cell] (1.95,0.25) rectangle (4.55,0.85);
				\node at (3.25,0.55) {$O_3$};
				\node[panelnote] at (2.5,-0.95) {some cell inside $X$;\\ no cell inside $\overline X$};
			\end{scope}
			
			\begin{scope}[shift={(0,0)}]
				\node[font=\bfseries] at (2.5,3.05) {$\mathrm{AD}$: Adverse Determination};
				\fill[blue!9] (0,0) rectangle (2.4,2.7);
				\fill[red!9] (2.4,0) rectangle (5,2.7);
				\draw[thick] (0,0) rectangle (5,2.7);
				\draw[boundary] (2.4,0) -- (2.4,2.7);
				\node at (1.2,2.45) {$X$};
				\node at (3.7,2.45) {$\overline X$};
				\draw[cell] (3.15,0.45) rectangle (4.55,1.25);
				\node at (3.85,0.85) {$O_1$};
				\draw[cell] (1.25,1.25) rectangle (3.45,1.95);
				\node at (2.35,1.60) {$O_2$};
				\draw[cell] (0.45,0.25) rectangle (2.75,0.85);
				\node at (1.60,0.55) {$O_3$};
				\node[panelnote] at (2.5,-0.95) {some cell inside $\overline X$;\\ no cell inside $X$};
			\end{scope}
			
			\begin{scope}[shift={(7.0,0)}]
				\node[font=\bfseries] at (2.5,3.05) {$\mathrm{FI}$: Full Inability};
				\fill[blue!9] (0,0) rectangle (2.4,2.7);
				\fill[red!9] (2.4,0) rectangle (5,2.7);
				\draw[thick] (0,0) rectangle (5,2.7);
				\draw[boundary] (2.4,0) -- (2.4,2.7);
				\node at (1.2,2.45) {$X$};
				\node at (3.7,2.45) {$\overline X$};
				\draw[cell] (1.25,0.35) rectangle (3.55,0.95);
				\node at (2.40,0.65) {$O_1$};
				\draw[cell] (1.65,1.15) rectangle (4.35,1.75);
				\node at (3.00,1.45) {$O_2$};
				\draw[cell] (0.75,1.85) rectangle (3.10,2.35);
				\node at (1.92,2.10) {$O_3$};
				\node[panelnote] at (2.5,-0.95) {every cell crosses\\ the boundary};
			\end{scope}
			
			\node[
			draw=gray!60,
			rounded corners=3pt,
			fill=gray!5,
			align=center,
			text width=11.6cm,
			font=\footnotesize
			] at (6.0,-2.65)
			{Each green cell $O_k$ represents the outcome cell $O_w(s_C^k)$ for some strategy $s_C^k$ of coalition $C$.};
			
		\end{tikzpicture}
		\caption{Strategy-cell geometric interpretation. Full Inability manifests as a universal boundary-crossing condition: every $C$-strategy cell intersects both $X$ and $\overline X$.}
		\label{fig:strategy-cell-categories}
	\end{figure}
	
	\begin{proposition}[Strategy-Cell Characterization]
		\label{prop:strategy-cell-characterization}
		Assume $E_w$ is induced by a strategic game form. For every $C \subseteq N$ and $X \subseteq W$:
		\begin{enumerate}[label=(\arabic*)]
			\item $X \in \Rfc^w(C)$ iff there exist $O_1,O_2 \in \mathcal{O}_C^w$ such that
			\[
			O_1 \subseteq X
			\quad\text{and}\quad
			O_2 \subseteq \overline{X}.
			\]
			
			\item $X \in \Rpd^w(C)$ iff there exists $O \in \mathcal{O}_C^w$ such that
			\[
			O \subseteq X,
			\]
			and for every $O \in \mathcal{O}_C^w$,
			\[
			O \cap X \neq \emptyset.
			\]
			
			\item $X \in \Rad^w(C)$ iff there exists $O \in \mathcal{O}_C^w$ such that
			\[
			O \subseteq \overline{X},
			\]
			and for every $O \in \mathcal{O}_C^w$,
			\[
			O \cap \overline{X} \neq \emptyset.
			\]
			
			\item $X \in \Rfi^w(C)$ iff for every $O \in \mathcal{O}_C^w$,
			\[
			O \cap X \neq \emptyset
			\quad\text{and}\quad
			O \cap \overline{X} \neq \emptyset.
			\]
		\end{enumerate}
	\end{proposition}
	
	\begin{proof}
		By game-form semantics,
		\[
		X \in E_w(C)
		\quad\Longleftrightarrow\quad
		\exists O \in \mathcal{O}_C^w \text{ such that } O \subseteq X.
		\]
		Since every outcome cell $O \in \mathcal{O}_C^w$ is non-empty, we have
		\[
		X \notin E_w(C)
		\quad\Longleftrightarrow\quad
		\forall O \in \mathcal{O}_C^w,\; O \nsubseteq X
		\quad\Longleftrightarrow\quad
		\forall O \in \mathcal{O}_C^w,\; O \cap \overline{X} \neq \emptyset.
		\]
		Similarly,
		\[
		\overline{X} \in E_w(C)
		\quad\Longleftrightarrow\quad
		\exists O \in \mathcal{O}_C^w \text{ such that } O \subseteq \overline{X},
		\]
		and
		\[
		\overline{X} \notin E_w(C)
		\quad\Longleftrightarrow\quad
		\forall O \in \mathcal{O}_C^w,\; O \cap X \neq \emptyset.
		\]
		
		For (1), $X \in \Rfc^w(C) = E_w(C) \cap E_w^*(C)$ means $X \in E_w(C)$ and $\overline{X} \in E_w(C)$. By the above equivalences, this holds iff there exist $O_1, O_2 \in \mathcal{O}_C^w$ with $O_1 \subseteq X$ and $O_2 \subseteq \overline{X}$.
		
		For (2), $X \in \Rpd^w(C) = E_w(C) \setminus E_w^*(C)$ means $X \in E_w(C)$ and $\overline{X} \notin E_w(C)$. This holds iff there exists $O \in \mathcal{O}_C^w$ with $O \subseteq X$, and for every $O \in \mathcal{O}_C^w$, $O \cap X \neq \emptyset$.
		
		For (3), by symmetry with (2), replacing $X$ by $\overline{X}$.
		
		For (4), $X \in \Rfi^w(C) = \powerset(W) \setminus (E_w(C) \cup E_w^*(C))$ means $X \notin E_w(C)$ and $\overline{X} \notin E_w(C)$. By the above equivalences, this holds iff for every $O \in \mathcal{O}_C^w$, $O \cap \overline{X} \neq \emptyset$ and $O \cap X \neq \emptyset$.
	\end{proof}
	
	The final clause characterizes Full Inability as a \emph{universal boundary-crossing condition}: no strategy of $C$ is fully contained in $X$ or in $\overline{X}$.

	\subsection{The Convexity Theorem}
	\label{subsec:convexity}
	
	A subset $S \subseteq \powerset(W)$ is \emph{order-convex} if whenever $X,Z \in S$ and
	\[
	X \subseteq Y \subseteq Z,
	\]
	we have $Y \in S$.
	
	\begin{theorem}[Convexity of Power Regions]
		\label{thm:convexity}
		In any playable coalition model, for every state $w$ and coalition $C \subseteq N$, the regions
		\[
		\Rfc^{w}(C),
		\quad
		\Rpd^{w}(C),
		\quad
		\Rad^{w}(C),
		\quad
		\Rfi^{w}(C)
		\]
		are order-convex in $(\powerset(W),\subseteq)$.
	\end{theorem}
	
	\begin{proof}
		Fix $w$ and $C$, and let
		\[
		E=E_w(C)
		\quad\text{and}\quad
		I=E_w^*(C).
		\]
		By Lemma~\ref{lem:lattice-polarity}, $E$ is upward closed and $I$ is downward closed.
		
		For $\Rfc^w(C)=E\cap I$, suppose $X,Z \in E\cap I$ and
		\[
		X \subseteq Y \subseteq Z.
		\]
		Since $X \in E$ and $X \subseteq Y$, upward closure gives $Y \in E$. Since $Z \in I$ and $Y \subseteq Z$, downward closure gives $Y \in I$. Hence $Y \in E\cap I$.
		
		For $\Rpd^w(C)=E\setminus I$, suppose $X,Z \in E\setminus I$ and
		\[
		X \subseteq Y \subseteq Z.
		\]
		Since $X \in E$ and $X \subseteq Y$, upward closure gives $Y \in E$. If $Y \in I$, then $X \in I$ by downward closure, contradicting $X \notin I$. Thus $Y \notin I$, and so $Y \in E\setminus I$.
		
		For $\Rad^w(C)=I\setminus E$, suppose $X,Z \in I\setminus E$ and
		\[
		X \subseteq Y \subseteq Z.
		\]
		Since $Z \in I$ and $Y \subseteq Z$, downward closure gives $Y \in I$. If $Y \in E$, then $Z \in E$ by upward closure, contradicting $Z \notin E$. Thus $Y \notin E$, and so $Y \in I\setminus E$.
		
		For
		\[
		\Rfi^w(C)=\powerset(W)\setminus(E\cup I),
		\]
		suppose $X,Z \notin E\cup I$ and
		\[
		X \subseteq Y \subseteq Z.
		\]
		If $Y \in E$, then $Z \in E$ by upward closure, contradicting $Z \notin E$. If $Y \in I$, then $X \in I$ by downward closure, contradicting $X \notin I$. Therefore $Y \notin E\cup I$, so $Y \in \Rfi^w(C)$.
	\end{proof}
	
	\begin{remark}
		The proof uses only outcome monotonicity, namely the upward closure of $E_w(C)$. The convexity result holds for any effectivity model whose effectivity regions are outcome-monotone.
	\end{remark}
	
	\subsection{Verification Significance: Contiguity and Safety Guardrails}
	\label{subsec:significance}
	
	The convexity of $\Rfi^{w}(C)$ underwrites \emph{Strategic Contiguity}:
	\[
	X,Z \in \Rfi^{w}(C)
	\text{ and }
	X \subseteq Y \subseteq Z
	\quad\Longrightarrow\quad
	Y \in \Rfi^{w}(C).
	\]
	Thus Full Inability is stable across refinement intervals.
	
	In verification terms, suppose $X$ is a stronger or more fine-grained specification and $Z$ is a weaker or more abstract specification, with $X \subseteq Z$. If a coalition has Full Inability with respect to both $X$ and $Z$, then it has Full Inability with respect to every intermediate specification $Y$ satisfying
	\[
	X \subseteq Y \subseteq Z.
	\]
	Hence if both a fine-grained safety condition and a coarser abstraction lie beyond a coalition's deterministic control, every intermediate specification inherits the same Full Inability guarantee.
	
	This provides a lattice-theoretic safety guardrail: once the endpoints of a specification interval are certified as uncontrollable by a coalition in both directions, no intermediate weakening or strengthening within that interval can accidentally restore deterministic control.
	
	\subsection{Coalitional Dynamics and the Migration of Power}
	\label{subsec:monotonicity-grand}
	
	The transition across power regions is constrained by coalitional growth. If $C \subseteq D$, then coalition monotonicity gives
	\[
	E_w(C) \subseteq E_w(D).
	\]
	Thus larger coalitions can inherit abilities of smaller coalitions, while smaller coalitions inherit inability from larger coalitions.
	
	\begin{proposition}[Coalitional Shift]
		\label{prop:coalitional-shift}
		Let $\Mod$ be a playable coalition model, $w \in W$, $\varphi \in \Lang_{\CL}$, and $C \subseteq D \subseteq N$. Then:
		\begin{enumerate}[label=(\arabic*)]
			\item Full Inability is anti-monotone:
			\[
			\Mod,w \models \FI_D(\varphi)
			\quad\Longrightarrow\quad
			\Mod,w \models \FI_C(\varphi).
			\]
			
			\item Full Control is monotone:
			\[
			\Mod,w \models \FC_C(\varphi)
			\quad\Longrightarrow\quad
			\Mod,w \models \FC_D(\varphi).
			\]
		\end{enumerate}
	\end{proposition}
	
	\begin{proof}
		By Coalition Monotonicity (Lemma~\ref{lem:coalition-monotonicity}), $E_w(C) \subseteq E_w(D)$.
		
		For~(1), assume $\Mod,w \models \FI_D(\varphi)$. Then
		\[
		\Mod,w \not\models \Eff{D}\varphi
		\quad\text{and}\quad
		\Mod,w \not\models \Eff{D}\neg\varphi.
		\]
		If $C$ could enforce $\varphi$, then $D$ could enforce $\varphi$ by coalition monotonicity, contradiction. Likewise, if $C$ could enforce $\neg\varphi$, then $D$ could enforce $\neg\varphi$, contradiction. Hence
		\[
		\Mod,w \not\models \Eff{C}\varphi
		\quad\text{and}\quad
		\Mod,w \not\models \Eff{C}\neg\varphi,
		\]
		so $\Mod,w \models \FI_C(\varphi)$.
		
		For~(2), assume $\Mod,w \models \FC_C(\varphi)$. Then
		\[
		\Mod,w \models \Eff{C}\varphi
		\quad\text{and}\quad
		\Mod,w \models \Eff{C}\neg\varphi.
		\]
		By coalition monotonicity, $D$ inherits both enforcement capabilities:
		\[
		\Mod,w \models \Eff{D}\varphi
		\quad\text{and}\quad
		\Mod,w \models \Eff{D}\neg\varphi.
		\]
		Therefore $\Mod,w \models \FC_D(\varphi)$.
	\end{proof}
	
	Expanding coalitions may escape $\Rfi$ toward one-sided or two-sided determination, while subcoalitions inherit Full Inability downward. Conversely, Full Control propagates upward along coalition inclusion.
	
	\subsection{The Inability Threshold and $k$-Robustness}
	\label{subsec:thresholds}
	
	The anti-monotonicity of Full Inability allows us to quantify resilience against collusion. Since larger coalitions may acquire abilities unavailable to their subcoalitions, the relevant boundary is the family of minimal coalitions that escape Full Inability.
	
	\begin{definition}[Inability Threshold]
		\label{def:inability-threshold}
		Let $\Mod$ be a playable coalition model, $w \in W$, and $\varphi \in \Lang_{\CL}$. The inability threshold is
		\[
		\mathcal{C}^{w}_{\min}(\varphi)
		\defeq
		\left\{
		C \subseteq N
		\;\middle|\;
		\Mod,w \not\models \FI_C(\varphi)
		\text{ and }
		\forall D \subset C,\; \Mod,w \models \FI_D(\varphi)
		\right\}.
		\]
	\end{definition}
	
	Thus $\mathcal{C}^{w}_{\min}(\varphi)$ collects the inclusion-minimal coalitions that can determine at least one side of the issue $\varphi$. Escaping Full Inability means satisfying
	\[
	\neg\FI_C(\varphi),
	\]
	equivalently,
	\[
	\Eff{C}\varphi
	\lor
	\Eff{C}\neg\varphi.
	\]
	It does not necessarily mean obtaining Full Control.
	
	\begin{proposition}[Antichain Property]
		\label{prop:threshold-antichain}
		The family $\mathcal{C}^{w}_{\min}(\varphi)$ is an antichain under inclusion.
	\end{proposition}
	
	\begin{proof}
		Suppose $C,D \in \mathcal{C}^{w}_{\min}(\varphi)$ and $C \subset D$. Since $D$ is minimal among coalitions that fail to satisfy Full Inability, every proper subcoalition of $D$ satisfies $\FI(\varphi)$. In particular,
		\[
		\Mod,w \models \FI_C(\varphi).
		\]
		This contradicts $C \in \mathcal{C}^{w}_{\min}(\varphi)$, which requires
		\[
		\Mod,w \not\models \FI_C(\varphi).
		\]
		Hence no two distinct members of $\mathcal{C}^{w}_{\min}(\varphi)$ are comparable by inclusion.
	\end{proof}
	
	\begin{definition}[Robustness Degree]
		\label{def:robustness-degree}
		The robustness degree is
		\[
		\mathrm{Robustness}^{w}(\varphi)
		\defeq
		\begin{cases}
			\min\{|C| : C \in \mathcal{C}^{w}_{\min}(\varphi)\},
			& \text{if } \mathcal{C}^{w}_{\min}(\varphi) \neq \emptyset, \\[2pt]
			\infty,
			& \text{otherwise}.
		\end{cases}
		\]
		A system is \emph{$k$-robust with respect to $\varphi$ at $w$} if
		\[
		\mathrm{Robustness}^{w}(\varphi)>k.
		\]
	\end{definition}
	
	Intuitively, $k$-robustness means that no coalition of size at most $k$ can escape Full Inability with respect to $\varphi$. This provides a quantitative measure of strategic neutralization: the higher the robustness degree, the larger the coalition required to gain any deterministic control over $\varphi$.
	
	\begin{proposition}[Robustness Characterization]
		\label{prop:robustness-characterization}
		Assume $\Mod$ is playable. The system is $k$-robust with respect to $\varphi$ at $w$ iff every coalition $C \subseteq N$ with $|C|\leq k$ satisfies
		\[
		\Mod,w \models \FI_C(\varphi).
		\]
	\end{proposition}
	
	\begin{proof}
		By Proposition~\ref{prop:coalitional-shift}, Full Inability is anti-monotone with respect to coalition inclusion. Equivalently, failure of Full Inability is monotone upward: if
		\[
		\Mod,w \not\models \FI_C(\varphi)
		\quad\text{and}\quad
		C \subseteq D,
		\]
		then
		\[
		\Mod,w \not\models \FI_D(\varphi).
		\]
		Indeed, if $C$ can enforce $\varphi$ or $\neg\varphi$, then every larger coalition $D$ can enforce the same side by coalition monotonicity.
		
		Since $N$ is finite, every coalition failing Full Inability contains an inclusion-minimal subcoalition failing Full Inability. Therefore $\mathrm{Robustness}^{w}(\varphi)$ is precisely the least size of a coalition that escapes Full Inability.
		
		Hence
		\[
		\mathrm{Robustness}^{w}(\varphi)>k
		\]
		holds iff no coalition of size at most $k$ fails Full Inability. Equivalently, every coalition $C \subseteq N$ with $|C|\leq k$ satisfies
		\[
		\Mod,w \models \FI_C(\varphi).
		\]
	\end{proof}
	
	Thus $k$-robustness gives a verifiable lower bound on the coalition size required to escape Full Inability. Importantly, escaping Full Inability need not mean obtaining Full Control: the minimal coalition may achieve only Positive Determination or Adverse Determination.

	\section{The Logic \texorpdfstring{\(\CLFI\)}{CL\^{}FI}: Axiomatization and Meta-Theory}
	\label{sec:axiomatization}
	
	This section formalizes the logic of Full Inability by providing a proof system with an explicit \(\FI\)-modality. The resulting system, denoted \(\CLFI\), is a definitional extension of standard Coalition Logic. The central thesis is that \(\FI_C(\varphi)\) does not increase the expressive power of the base language beyond the ordinary effectivity modality \(\Eff{C}\). Rather, it packages a strategically significant negative configuration into a single, proof-theoretically tractable operator.
	
	\subsection{Language and Proof System}
	\label{subsec:clfi-lang}
	
	The extended language \(\Lang_{\CLFI}\) builds upon the standard language \(\Lang_{\CL}\) of Coalition Logic by treating the Full Inability operator as a primitive modality.
	
	\begin{definition}[Language \(\Lang_{\CLFI}\)]
		\label{def:lang-clfi}
		Formulas of \(\Lang_{\CLFI}\) are recursively defined by the grammar
		\[
		\varphi ::= p \mid \neg\varphi \mid (\varphi \wedge \psi) \mid \Eff{C}\varphi \mid \FI_C(\varphi),
		\]
		where \(p \in \Prop\) and \(C \subseteq N\).
		
		The Boolean connectives \(\vee\), \(\to\), and \(\leftrightarrow\) are defined as standard abbreviations.
	\end{definition}
	
	The semantics of \(\FI_C(\varphi)\) is determined by the simultaneous absence of both positive and negative enforceability.
	
	\begin{definition}[Semantics of Full Inability]
		\label{def:semantics-fi}
		Let \(\Mod = (W, E, V)\) be a playable coalition model, let \(w \in W\), and let \(C \subseteq N\). Then
		\[
		\Mod, w \models \FI_C(\varphi)
		\]
		if and only if
		\[
		\sem{\varphi}_{\Mod} \notin E_w(C)
		\quad\text{and}\quad
		\sem{\neg\varphi}_{\Mod} \notin E_w(C).
		\]
		Equivalently,
		\[
		\Mod, w \models \FI_C(\varphi)
		\quad\Longleftrightarrow\quad
		\Mod, w \models \neg\Eff{C}\varphi \wedge \neg\Eff{C}\neg\varphi.
		\]
	\end{definition}
	
	The proof system \(\CLFI\) inherits all axiom schemes and inference rules of standard Coalition Logic---namely, the propositional tautologies (\textbf{PL}), the bottom axiom (\textbf{\(\bot\)}), outcome monotonicity (\textbf{M}), superadditivity (\textbf{S}), \(N\)-maximality (\textbf{N}), and the rule of replacement of equivalents (\textbf{RE})~\cite{Pauly02}. These schemes and rules are understood over the extended language \(\Lang_{\CLFI}\). In particular, \textbf{RE} permits replacement of provably equivalent \(\Lang_{\CLFI}\)-formulas under \(\Eff{C}\), and outcome monotonicity \textbf{M} applies to implications between \(\Lang_{\CLFI}\)-formulas.
	
	The system is augmented with the following definitional axiom for Full Inability:
	\begin{equation}
		\tag{\(\mathrm{Def}\text{-}\FI\)}
		\FI_C(\varphi)
		\leftrightarrow
		\bigl(
		\neg\Eff{C}\varphi
		\wedge
		\neg\Eff{C}\neg\varphi
		\bigr).
	\end{equation}
	
	Consequently, \(\FI_C(\varphi)\) is proof-theoretically eliminable in favor of the ordinary coalition modality \(\Eff{C}\). The extension is therefore definitional: while the new operator introduces no additional semantic primitive, it allows the logic to directly name and manipulate a strategically significant configuration.
	
	\subsection{Derived Structural Principles}
	\label{subsec:clfi-axioms}
	
	The system \(\CLFI\) derives characteristic structural principles governing Full Inability. These derivations formally establish that \(\FI\) behaves as a stable negative modality generated by bilateral failures of enforceability.
	
	\begin{proposition}[Derived Principles]
		\label{prop:derived-principles}
		The following principles are derivable in \(\CLFI\).
		\begin{enumerate}[label=(\roman*)]
			\item \textbf{Negation Invariance}:
			\[
			\vdash_{\CLFI} \FI_C(\varphi) \leftrightarrow \FI_C(\neg\varphi).
			\]
			
			\item \textbf{Anti-Monotonicity in Coalitions}: For all \(C \subseteq D\),
			\[
			\vdash_{\CLFI} \FI_D(\varphi) \to \FI_C(\varphi).
			\]
			
			\item \textbf{Grand Coalition Determination}: If
			\[
			\vdash_{\CL} \Eff{N}\varphi \vee \Eff{N}\neg\varphi,
			\]
			then
			\[
			\vdash_{\CLFI} \neg\FI_N(\varphi).
			\]
			
			\item \textbf{Convexity Principle}: If
			\[
			\vdash_{\CLFI} \varphi \to \chi
			\quad\text{and}\quad
			\vdash_{\CLFI} \chi \to \psi,
			\]
			then
			\[
			\vdash_{\CLFI}
			\bigl(
			\FI_C(\varphi) \wedge \FI_C(\psi)
			\bigr)
			\to
			\FI_C(\chi).
			\]
		\end{enumerate}
	\end{proposition}
	
	\begin{proof}
		We provide derivations for each principle.
		
		\begin{enumerate}[label=(\roman*)]
			\item \textbf{Negation Invariance.}
			By \(\mathrm{Def}\text{-}\FI\),
			\[
			\FI_C(\varphi)
			\leftrightarrow
			\bigl(
			\neg\Eff{C}\varphi
			\wedge
			\neg\Eff{C}\neg\varphi
			\bigr).
			\]
			Similarly,
			\[
			\FI_C(\neg\varphi)
			\leftrightarrow
			\bigl(
			\neg\Eff{C}\neg\varphi
			\wedge
			\neg\Eff{C}\neg\neg\varphi
			\bigr).
			\]
			By classical propositional logic,
			\[
			\vdash_{\CLFI} \neg\neg\varphi \leftrightarrow \varphi.
			\]
			Applying \textbf{RE} for \(\Eff{C}\), we obtain
			\[
			\vdash_{\CLFI}
			\Eff{C}\neg\neg\varphi
			\leftrightarrow
			\Eff{C}\varphi.
			\]
			Therefore,
			\[
			\vdash_{\CLFI}
			\neg\Eff{C}\neg\neg\varphi
			\leftrightarrow
			\neg\Eff{C}\varphi.
			\]
			Using this equivalence together with commutativity and associativity of conjunction, we derive
			\[
			\vdash_{\CLFI}
			\FI_C(\varphi)
			\leftrightarrow
			\FI_C(\neg\varphi).
			\]
			
			\item \textbf{Anti-Monotonicity in Coalitions.}
			Let \(C \subseteq D\). By the derivable coalition monotonicity principle of Coalition Logic, applied over the extended language,
			\[
			\vdash_{\CLFI} \Eff{C}\varphi \to \Eff{D}\varphi.
			\]
			By propositional contraposition,
			\[
			\vdash_{\CLFI} \neg\Eff{D}\varphi \to \neg\Eff{C}\varphi.
			\]
			Applying the same reasoning to \(\neg\varphi\) yields
			\[
			\vdash_{\CLFI} \Eff{C}\neg\varphi \to \Eff{D}\neg\varphi,
			\]
			and therefore
			\[
			\vdash_{\CLFI} \neg\Eff{D}\neg\varphi \to \neg\Eff{C}\neg\varphi.
			\]
			Conjoining these implications gives
			\[
			\vdash_{\CLFI}
			\bigl(
			\neg\Eff{D}\varphi
			\wedge
			\neg\Eff{D}\neg\varphi
			\bigr)
			\to
			\bigl(
			\neg\Eff{C}\varphi
			\wedge
			\neg\Eff{C}\neg\varphi
			\bigr).
			\]
			Folding both sides via \(\mathrm{Def}\text{-}\FI\), we obtain
			\[
			\vdash_{\CLFI}
			\FI_D(\varphi)
			\to
			\FI_C(\varphi).
			\]
			
			\item \textbf{Grand Coalition Determination.}
			Assume
			\[
			\vdash_{\CL} \Eff{N}\varphi \vee \Eff{N}\neg\varphi.
			\]
			Since every theorem of \(\CL\) is derivable in \(\CLFI\), we have
			\[
			\vdash_{\CLFI} \Eff{N}\varphi \vee \Eff{N}\neg\varphi.
			\]
			By propositional logic,
			\[
			\vdash_{\CLFI}
			\bigl(
			\Eff{N}\varphi
			\vee
			\Eff{N}\neg\varphi
			\bigr)
			\to
			\neg
			\bigl(
			\neg\Eff{N}\varphi
			\wedge
			\neg\Eff{N}\neg\varphi
			\bigr).
			\]
			By modus ponens,
			\[
			\vdash_{\CLFI}
			\neg
			\bigl(
			\neg\Eff{N}\varphi
			\wedge
			\neg\Eff{N}\neg\varphi
			\bigr).
			\]
			By \(\mathrm{Def}\text{-}\FI\),
			\[
			\vdash_{\CLFI}
			\FI_N(\varphi)
			\leftrightarrow
			\bigl(
			\neg\Eff{N}\varphi
			\wedge
			\neg\Eff{N}\neg\varphi
			\bigr).
			\]
			Hence, by propositional reasoning,
			\[
			\vdash_{\CLFI} \neg\FI_N(\varphi).
			\]
			
			\item \textbf{Convexity Principle.}
			Assume
			\[
			\vdash_{\CLFI} \varphi \to \chi
			\quad\text{and}\quad
			\vdash_{\CLFI} \chi \to \psi.
			\]
			
			First, by outcome monotonicity \textbf{M} applied to the second premise,
			\[
			\vdash_{\CLFI} \Eff{C}\chi \to \Eff{C}\psi.
			\]
			By propositional contraposition,
			\[
			\vdash_{\CLFI} \neg\Eff{C}\psi \to \neg\Eff{C}\chi.
			\]
			From \(\mathrm{Def}\text{-}\FI\),
			\[
			\vdash_{\CLFI} \FI_C(\psi) \to \neg\Eff{C}\psi.
			\]
			Therefore,
			\[
			\vdash_{\CLFI} \FI_C(\psi) \to \neg\Eff{C}\chi.
			\]
			
			Second, from \(\vdash_{\CLFI} \varphi \to \chi\), propositional contraposition gives
			\[
			\vdash_{\CLFI} \neg\chi \to \neg\varphi.
			\]
			By outcome monotonicity \textbf{M},
			\[
			\vdash_{\CLFI} \Eff{C}\neg\chi \to \Eff{C}\neg\varphi.
			\]
			By propositional contraposition,
			\[
			\vdash_{\CLFI} \neg\Eff{C}\neg\varphi \to \neg\Eff{C}\neg\chi.
			\]
			From \(\mathrm{Def}\text{-}\FI\),
			\[
			\vdash_{\CLFI} \FI_C(\varphi) \to \neg\Eff{C}\neg\varphi.
			\]
			Therefore,
			\[
			\vdash_{\CLFI} \FI_C(\varphi) \to \neg\Eff{C}\neg\chi.
			\]
			
			Combining the two derived implications yields
			\[
			\vdash_{\CLFI}
			\bigl(
			\FI_C(\varphi)
			\wedge
			\FI_C(\psi)
			\bigr)
			\to
			\bigl(
			\neg\Eff{C}\chi
			\wedge
			\neg\Eff{C}\neg\chi
			\bigr).
			\]
			Folding the consequent via \(\mathrm{Def}\text{-}\FI\), we conclude
			\[
			\vdash_{\CLFI}
			\bigl(
			\FI_C(\varphi)
			\wedge
			\FI_C(\psi)
			\bigr)
			\to
			\FI_C(\chi).
			\]
		\end{enumerate}
	\end{proof}
	
	The Convexity Principle is the proof-theoretic counterpart of the semantic order-convexity established in Theorem~\ref{thm:convexity}. If \(\chi\) is logically sandwiched between \(\varphi\) and \(\psi\), then Full Inability at the two extremes guarantees Full Inability for the intermediate specification.

	\subsection{Elimination Translation}
	\label{subsec:clfi-translation}
	
	To establish the meta-logical properties of \(\CLFI\), we define an elimination translation
	\[
	\tr \colon \Lang_{\CLFI} \to \Lang_{\CL}.
	\]
	This mapping removes every occurrence of \(\FI\) by expanding it into its defining \(\Eff{C}\)-formula.
	
	\begin{definition}[Elimination Translation]
		\label{def:elimination-translation}
		The translation \(\tr \colon \Lang_{\CLFI} \to \Lang_{\CL}\) is recursively defined by:
		\begin{align*}
			\tr(p) &= p, \\
			\tr(\neg\varphi) &= \neg \tr(\varphi), \\
			\tr(\varphi \wedge \psi) &= \tr(\varphi) \wedge \tr(\psi), \\
			\tr(\Eff{C}\varphi) &= \Eff{C} \tr(\varphi), \\
			\tr(\FI_C(\varphi)) &=
			\neg\Eff{C} \tr(\varphi)
			\wedge
			\neg\Eff{C} \neg \tr(\varphi).
		\end{align*}
	\end{definition}
	
	\begin{lemma}[Truth Preservation]
		\label{lem:truth-preservation}
		For every playable coalition model \(\Mod\), every state \(w\), and every formula \(\varphi \in \Lang_{\CLFI}\),
		\[
		\Mod, w \models \varphi
		\quad\Longleftrightarrow\quad
		\Mod, w \models \tr(\varphi).
		\]
	\end{lemma}
	
	\begin{proof}
		We proceed by structural induction on \(\varphi\). The atomic and Boolean cases are immediate.
		
		For the effectivity modality, suppose \(\varphi = \Eff{C}\psi\). By the induction hypothesis,
		\[
		\sem{\psi}_{\Mod} = \sem{\tr(\psi)}_{\Mod}.
		\]
		Thus,
		\[
		\Mod, w \models \Eff{C}\psi
		\;\Longleftrightarrow\;
		\sem{\psi}_{\Mod} \in E_w(C)
		\;\Longleftrightarrow\;
		\sem{\tr(\psi)}_{\Mod} \in E_w(C)
		\;\Longleftrightarrow\;
		\Mod, w \models \Eff{C}\tr(\psi).
		\]
		Since \(\tr(\Eff{C}\psi) = \Eff{C}\tr(\psi)\), we obtain
		\[
		\Mod, w \models \Eff{C}\psi
		\;\Longleftrightarrow\;
		\Mod, w \models \tr(\Eff{C}\psi).
		\]
		
		For the Full Inability operator, suppose \(\varphi = \FI_C(\psi)\). By Definition~\ref{def:semantics-fi},
		\[
		\Mod, w \models \FI_C(\psi)
		\]
		if and only if
		\[
		\sem{\psi}_{\Mod} \notin E_w(C)
		\quad\text{and}\quad
		\sem{\neg\psi}_{\Mod} \notin E_w(C).
		\]
		By the induction hypothesis,
		\[
		\sem{\psi}_{\Mod} = \sem{\tr(\psi)}_{\Mod}.
		\]
		Since Boolean negation is interpreted as set-theoretic complementation,
		\[
		\sem{\neg\psi}_{\Mod}
		=
		W\setminus\sem{\psi}_{\Mod}
		=
		W\setminus\sem{\tr(\psi)}_{\Mod}
		=
		\sem{\neg\tr(\psi)}_{\Mod}.
		\]
		Substituting these equalities gives
		\[
		\sem{\tr(\psi)}_{\Mod} \notin E_w(C)
		\quad\text{and}\quad
		\sem{\neg \tr(\psi)}_{\Mod} \notin E_w(C),
		\]
		which is equivalent to
		\[
		\Mod, w \models
		\neg\Eff{C}\tr(\psi)
		\wedge
		\neg\Eff{C}\neg \tr(\psi).
		\]
		Since \(\tr(\FI_C(\psi)) = \neg\Eff{C}\tr(\psi) \wedge \neg\Eff{C}\neg \tr(\psi)\), we obtain
		\[
		\Mod, w \models \FI_C(\psi)
		\;\Longleftrightarrow\;
		\Mod, w \models \tr(\FI_C(\psi)).
		\]
		This completes the induction.
	\end{proof}
	
	\begin{lemma}[Provable Equivalence under Translation]
		\label{lem:derivable-elimination}
		For every formula \(\varphi \in \Lang_{\CLFI}\),
		\[
		\vdash_{\CLFI} \varphi \leftrightarrow \tr(\varphi).
		\]
	\end{lemma}
	
	\begin{proof}
		We proceed by structural induction on \(\varphi\). The atomic case is trivial, and the Boolean cases follow from classical propositional logic together with the induction hypothesis.
		
		For \(\varphi = \Eff{C}\psi\), the induction hypothesis gives
		\[
		\vdash_{\CLFI} \psi \leftrightarrow \tr(\psi).
		\]
		Applying \textbf{RE}, we obtain
		\[
		\vdash_{\CLFI}
		\Eff{C}\psi
		\leftrightarrow
		\Eff{C}\tr(\psi).
		\]
		Since \(\tr(\Eff{C}\psi) = \Eff{C}\tr(\psi)\), the claim follows.
		
		For \(\varphi = \FI_C(\psi)\), the axiom \(\mathrm{Def}\text{-}\FI\) gives
		\[
		\vdash_{\CLFI}
		\FI_C(\psi)
		\leftrightarrow
		\bigl(
		\neg\Eff{C}\psi
		\wedge
		\neg\Eff{C}\neg\psi
		\bigr).
		\]
		The induction hypothesis gives
		\[
		\vdash_{\CLFI} \psi \leftrightarrow \tr(\psi).
		\]
		By \textbf{RE},
		\[
		\vdash_{\CLFI}
		\Eff{C}\psi
		\leftrightarrow
		\Eff{C}\tr(\psi).
		\]
		By propositional reasoning,
		\[
		\vdash_{\CLFI}
		\neg\psi
		\leftrightarrow
		\neg \tr(\psi).
		\]
		Another application of \textbf{RE} yields
		\[
		\vdash_{\CLFI}
		\Eff{C}\neg\psi
		\leftrightarrow
		\Eff{C}\neg \tr(\psi).
		\]
		Hence propositional reasoning gives
		\[
		\vdash_{\CLFI}
		\neg\Eff{C}\psi
		\leftrightarrow
		\neg\Eff{C}\tr(\psi),
		\]
		and
		\[
		\vdash_{\CLFI}
		\neg\Eff{C}\neg\psi
		\leftrightarrow
		\neg\Eff{C}\neg \tr(\psi).
		\]
		Substituting equivalents, we derive
		\[
		\vdash_{\CLFI}
		\FI_C(\psi)
		\leftrightarrow
		\bigl(
		\neg\Eff{C}\tr(\psi)
		\wedge
		\neg\Eff{C}\neg \tr(\psi)
		\bigr).
		\]
		Since the right-hand side is exactly \(\tr(\FI_C(\psi))\), we obtain
		\[
		\vdash_{\CLFI}
		\FI_C(\psi)
		\leftrightarrow
		\tr(\FI_C(\psi)).
		\]
	\end{proof}

	\subsection{Meta-Theorems}
	\label{subsec:clfi-metatheory}
	
	The elimination translation \(\tr\) establishes that \(\CLFI\) is a conservative definitional extension of standard Coalition Logic.
	
	\begin{theorem}[Soundness and Completeness]
		\label{thm:clfi-sound-complete}
		The system \(\CLFI\) is sound and complete with respect to the class of playable coalition models.
	\end{theorem}
	
	\begin{proof}
		\textbf{Soundness.}
		The axioms and rules inherited from \(\CL\) are sound over playable coalition models~\cite{Pauly02}. The rules \textbf{M} and \textbf{RE} remain sound over the extended language because every \(\Lang_{\CLFI}\)-formula denotes a well-defined truth set in every model. The additional axiom \(\mathrm{Def}\text{-}\FI\) is valid by Definition~\ref{def:semantics-fi}. Hence every theorem of \(\CLFI\) is valid.
		
		\medskip
		\noindent\textbf{Completeness.}
		Let \(\varphi \in \Lang_{\CLFI}\) be valid over all playable coalition models. By Lemma~\ref{lem:truth-preservation},
		\[
		\models \tr(\varphi).
		\]
		Since \(\tr(\varphi) \in \Lang_{\CL}\) and standard Coalition Logic is complete over playable models~\cite{Pauly02, Pauly01}, we obtain
		\[
		\vdash_{\CL} \tr(\varphi).
		\]
		All \(\CL\)-theorems are derivable in \(\CLFI\), so
		\[
		\vdash_{\CLFI} \tr(\varphi).
		\]
		By Lemma~\ref{lem:derivable-elimination},
		\[
		\vdash_{\CLFI} \varphi \leftrightarrow \tr(\varphi).
		\]
		Therefore, by propositional reasoning and modus ponens,
		\[
		\vdash_{\CLFI} \varphi.
		\]
		Thus \(\CLFI\) is complete.
	\end{proof}
	
	\begin{theorem}[Conservativity]
		\label{thm:clfi-conservativity}
		For every formula \(\varphi \in \Lang_{\CL}\),
		\[
		\vdash_{\CLFI} \varphi
		\quad\Longleftrightarrow\quad
		\vdash_{\CL} \varphi.
		\]
	\end{theorem}
	
	\begin{proof}
		The right-to-left direction holds because \(\CLFI\) contains all axioms and rules of \(\CL\).
		
		Conversely, assume \(\vdash_{\CLFI} \varphi\), where \(\varphi \in \Lang_{\CL}\). By soundness of \(\CLFI\),
		\[
		\models \varphi.
		\]
		Since \(\varphi\) belongs to the standard Coalition Logic language, completeness of \(\CL\) gives
		\[
		\vdash_{\CL} \varphi.
		\]
		Hence \(\CLFI\) is conservative over \(\CL\).
	\end{proof}
	
	\begin{theorem}[Decidability and Complexity]
		\label{thm:clfi-complexity}
		The satisfiability problem for \(\CLFI\) is \(\PSPACE\)-complete.
	\end{theorem}
	
	\begin{proof}
		\textbf{Lower bound.}
		Since \(\Lang_{\CL}\) is a syntactic fragment of \(\Lang_{\CLFI}\), every \(\CL\)-satisfiability instance is also a \(\CLFI\)-satisfiability instance. As \(\CL\)-satisfiability is \(\PSPACE\)-hard~\cite{Pauly02}, \(\CLFI\)-satisfiability is \(\PSPACE\)-hard.
		
		\medskip
		\noindent\textbf{Upper bound.}
		Satisfiability can be decided by the standard polynomial-space decision procedure for Coalition Logic, extended with a local unfolding rule for \(\FI\):
		\[
		\FI_C(\psi)
		\equiv
		\neg\Eff{C}\psi
		\wedge
		\neg\Eff{C}\neg\psi.
		\]
		Equivalently,
		\[
		\neg\FI_C(\psi)
		\equiv
		\Eff{C}\psi
		\vee
		\Eff{C}\neg\psi.
		\]
		
		Define the extended closure \(\mathrm{Cl}_{\FI}(\varphi)\) to contain the ordinary subformulas of \(\varphi\), their negations, and, for each occurrence of \(\FI_C(\psi)\), the formulas
		\[
		\Eff{C}\psi,\quad
		\Eff{C}\neg\psi,\quad
		\neg\Eff{C}\psi,\quad
		\neg\Eff{C}\neg\psi.
		\]
		The construction is performed over the syntax DAG (directed acyclic graph) of \(\varphi\), so subformulas are shared rather than copied. Hence
		\[
		|\mathrm{Cl}_{\FI}(\varphi)|
		\]
		is polynomial in \(|\varphi|\).
		
		The standard \(\PSPACE\) tableau or automata-based decision procedure for Coalition Logic can then be applied to this extended closure. Whenever a node contains \(\FI_C(\psi)\), the procedure treats it as requiring both
		\[
		\neg\Eff{C}\psi
		\quad\text{and}\quad
		\neg\Eff{C}\neg\psi.
		\]
		Whenever a node contains \(\neg\FI_C(\psi)\), the procedure treats it as requiring
		\[
		\Eff{C}\psi
		\quad\text{or}\quad
		\Eff{C}\neg\psi.
		\]
		These are local Boolean expansions and do not increase the space consumption beyond polynomial space. The effectivity constraints are then handled exactly as in the standard \(\CL\) decision procedure.
		
		Therefore, \(\CLFI\)-satisfiability is in \(\PSPACE\). Together with the lower bound, this establishes \(\PSPACE\)-completeness.
	\end{proof}
	
	\begin{remark}[Conceptual Value of the Extension]
		\label{rem:conceptual-value}
		Promoting \(\FI\) to an explicit modality does not introduce additional computational overhead or alter the expressive power of Coalition Logic. The value of \(\CLFI\) is conceptual and structural: it isolates the strategically significant case in which a coalition can enforce neither a specification nor its negation, providing direct proof-theoretic access to the notion of Full Inability.
	\end{remark}

	\section{Conclusion and Future Directions}
	\label{sec:conclusion}
	
	This paper has investigated the formal structure of coalitional inability within the framework of Coalition Logic. Building on recent work that elevated inability to an independent modality \cite{wang2026logic}, we have demonstrated that inability itself has internal structure. Rather than treating inability as a single undifferentiated concept---the simple negation of ability---we have shown that ability and inability together form a mathematically precise four-fold strategic spectrum. Within this spectrum, we developed a systematic logic of \emph{Full Inability} (\(\FI\)), which isolates the strongest form of coalitional limitation: a configuration in which a coalition can enforce neither a proposition nor its negation.
	
	\subsection{Summary of Contributions}
	\label{subsec:summary}
	
	The main contributions of this paper are organized along four dimensions.
	
	\begin{enumerate}[label=(\arabic*)]
		\item \textbf{The Four-Fold Strategic Spectrum.}
		We replaced the coarse binary distinction between ability and non-ability with a four-fold partition:
		\[
		\FC, \quad \PD, \quad \AD, \quad \FI.
		\]
		This spectrum provides an exhaustive and mutually exclusive classification of a coalition's strategic relation to any proposition. It explicitly separates the capacity to enforce a specific truth value from the stronger capacity to determine the issue in either direction. Full Control (\(\FC\)) and Full Inability (\(\FI\)) define the two extremes of the determination axis, whereas Positive Determination (\(\PD\)) and Adverse Determination (\(\AD\)) represent asymmetric, one-sided forms of strategic power.
		
		\item \textbf{Symmetry and Conditional Duality.}
		We established that the four-fold spectrum is governed by a Klein four-group (\(V_4\)) symmetry generated by propositional negation and coalition complementation, under the assumption of \(\alpha\)-duality. This algebraic perspective reveals that \(\FI\) is not merely a residual category but a structurally significant counterpart to \(\FC\). In particular, the invariance of \(\FI\) under propositional negation, as established in Proposition~\ref{prop:negation-symmetry}, confirms its role as a modality of pure non-determination: if a coalition is fully unable with respect to \(\varphi\), then it is equally fully unable with respect to \(\neg\varphi\).
		
		\item \textbf{Order-Convexity and Strategic Contiguity.}
		By mapping formulas to their truth sets within the powerset lattice \((\powerset(W), \subseteq)\), we proved the Convexity Theorem (Theorem~\ref{thm:convexity}) for the four power regions. This order-theoretic result establishes the stability of the strategic categories under interpolation in the subset order. In particular, if two outcome sets lie in the Full Inability region \(\Rfi^w(C)\), then every intermediate outcome set remains similarly outside both the enforceable region \(E_w(C)\) and the co-enforceable region \(E_w^*(C)\). This property, which we termed \emph{Strategic Contiguity}, supports interval-based verification of inability guarantees.
		
		\item \textbf{The Logic \(\CLFI\).}
		The proof system \(\CLFI\) provides a proof-theoretic treatment of Full Inability as an explicit modality. Through an elimination translation into standard Coalition Logic, we established soundness, completeness, and conservativity, as stated in Theorems~\ref{thm:clfi-sound-complete}--\ref{thm:clfi-conservativity}. The extension adds conceptual clarity without increasing expressive power or computational complexity: \(\CLFI\)-satisfiability remains \(\PSPACE\)-complete (Theorem~\ref{thm:clfi-complexity}).
	\end{enumerate}
	
	\subsection{Conceptual and Practical Significance}
	\label{subsec:discussion-significance}
	
	The explicit formalization of \(\FI\) bridges qualitative logical accounts of agency with fine-grained analyses of strategic dependence.
	
	\paragraph{Connection to Power Indices.}
	In social choice theory and cooperative game theory, classical power indices such as the Banzhaf and Shapley--Shubik indices measure the global structural influence of agents across voting mechanisms or coalitional games. By contrast, \(\FI_C(\varphi)\) is local, state-dependent, and proposition-sensitive. It identifies precisely when and where a coalition is strategically neutralized with respect to a specific issue at a specific state.
	
	This localized perspective refines the concept of a \emph{propositional dummy} (Definition~\ref{def:dummy}). An agent may possess substantial structural power over certain propositions while exhibiting Full Inability over others. Dummyhood is therefore not an immutable global property but an emergent characteristic relative to a specific proposition, state, and effectivity structure. The \(\CLFI\) framework is thus suitable for modeling context-sensitive forms of strategic irrelevance, localized dependence, and shifting distributions of influence.
	
	\paragraph{Applications to AI Safety.}
	In AI safety and formal verification, Full Inability provides a rigorous specification for certain forms of \emph{strategic containment}. Traditional safety requirements often take a negative form: verifying that an autonomous agent cannot force a harmful outcome. The \(\FI\) modality strengthens this negative requirement by requiring that the agent can force neither the critical event nor its complement. Thus \(\FI\) can function as an \emph{inability guardrail} for variables that should remain outside the unilateral control of a given agent or coalition.
	
	This interpretation should be applied with care. Full Inability does not by itself guarantee that a desirable safety condition will hold; rather, it guarantees that the specified coalition cannot determine the relevant proposition in either direction. In safety-critical settings, \(\FI_C(\varphi)\) is therefore most naturally combined with positive enforceability requirements for trusted supervisors, environmental constraints, or external verification mechanisms. Its value lies in formally certifying the absence of unilateral strategic control by the potentially unsafe coalition.
	
	The order-convexity of the Full Inability region \(\Rfi^w(C)\) can streamline verification. If two bounding specifications are certified to lie within \(\Rfi^w(C)\), then every intermediate specification inherits the same \(\FI\)-status. This allows verification algorithms to certify contiguous intervals of inability rather than checking each specification independently.
	
	\subsection{Future Directions}
	\label{subsec:future}
	
	The results of this paper suggest several directions for future research.
	
	\paragraph{Temporal Extensions.}
	A natural next step is to lift the four-fold spectrum into temporal logics of strategic ability, most notably Alternating-time Temporal Logic (\(\ATL\)). In a temporal setting, Full Inability can be defined relative to path specifications. If \(\coop{C}\) denotes the \(\ATL\) strategic modality, then for a temporal objective \(\Phi\) one may define
	\[
	\FI_C^{\ATL}(\Phi)
	\;\defeq\;
	\neg\coop{C}\Phi
	\;\wedge\;
	\neg\coop{C}\neg\Phi.
	\]
	For example, taking \(\Phi = \Box\varphi\) yields
	\[
	\FI_C^{\ATL}(\Box\varphi)
	\;\equiv\;
	\neg\coop{C}\Box\varphi
	\;\wedge\;
	\neg\coop{C}\Diamond\neg\varphi.
	\]
	The first conjunct states that \(C\) cannot enforce the invariant \(\Box\varphi\); the second states that \(C\) cannot force its temporal dual \(\Diamond\neg\varphi\). This extension would enable the study of persistent inability across the temporal evolution of a system. An important open question is whether temporal Full Inability admits fixed-point characterizations analogous to those in the modal \(\mu\)-calculus.
	
	\paragraph{Epistemic Extensions.}
	Another direction is the integration of \(\FI\) with epistemic logic. Formulas such as
	\[
	K_i \FI_C(\varphi)
	\]
	would express that agent \(i\) knows coalition \(C\) is fully unable to determine \(\varphi\). This enables reasoning about higher-order strategic awareness. Such a framework could formalize \emph{rational resignation}: when Full Inability becomes common knowledge, rational agents may abandon futile coordination. Conversely, it could model \emph{strategic deception}, where a coalition appears fully unable while retaining hidden enforceability. The interaction between verified inability, epistemic uncertainty \cite{DDP24}, and deceptive signaling merits systematic investigation.
	
	\paragraph{Causality and Responsibility.}
	Full Inability has natural connections to theories of causality and responsibility in multi-agent systems. In causal models \cite{Halpern16, Beckers24}, an agent or coalition is causally irrelevant to an outcome if its actions do not affect whether the outcome occurs. Full Inability provides a modal-logical counterpart: if \(\FI_C(\varphi)\) holds, then coalition \(C\) is strategically irrelevant to the truth of \(\varphi\). This suggests a formal bridge between effectivity-based and causality-based accounts of agency. Similarly, in frameworks for responsibility and accountability \cite{LMW23}, Full Inability may serve as a sufficient condition for exculpation: a coalition fully unable to determine an outcome cannot be held responsible for it. Integrating \(\FI\) with causal and deontic logics could yield a unified account of power, causation, and moral responsibility in strategic settings.
	
	\paragraph{Quantitative Refinements.}
	Although \(\FI\) is qualitative, it provides a foundation for quantitative analysis. For a model \(\Mod\), an agent \(i\), and a proposition \(\varphi\), one could define state sets corresponding to the four power categories:
	\begin{align*}
		S_{\FC}^{\Mod}(i,\varphi)
		&=
		\{\,w\in W \mid \sem{\varphi}_{\Mod}\in \Rfc^w(\{i\})\,\},\\
		S_{\PD}^{\Mod}(i,\varphi)
		&=
		\{\,w\in W \mid \sem{\varphi}_{\Mod}\in \Rpd^w(\{i\})\,\},\\
		S_{\AD}^{\Mod}(i,\varphi)
		&=
		\{\,w\in W \mid \sem{\varphi}_{\Mod}\in \Rad^w(\{i\})\,\},\\
		S_{\FI}^{\Mod}(i,\varphi)
		&=
		\{\,w\in W \mid \sem{\varphi}_{\Mod}\in \Rfi^w(\{i\})\,\}.
	\end{align*}
	This allows the definition of a \emph{strategic profile}
	\[
	\mathrm{Profile}_{\Mod}(i,\varphi)
	\;=\;
	\bigl(
	|S_{\FC}^{\Mod}(i,\varphi)|,\;
	|S_{\PD}^{\Mod}(i,\varphi)|,\;
	|S_{\AD}^{\Mod}(i,\varphi)|,\;
	|S_{\FI}^{\Mod}(i,\varphi)|
	\bigr),
	\]
	counting the states at which each power category holds. Such profiles constitute logical analogues of classical power indices. Comparing these state-based counts with cooperative game-theoretic measures may yield new connections between modal semantics and quantitative social choice theory.
	
	\paragraph{Bilattice and Galois-Theoretic Generalizations.}
	The Strategic Bilattice introduced in Section~\ref{subsec:strategic-bilattice} suggests a deeper algebraic development of the four-fold spectrum. Each strategic status of a coalition with respect to a proposition can be represented by the pair
	\[
	\bigl(
	\Eff{C}\varphi,\;
	\Eff{C}\neg\varphi
	\bigr)
	\in \{0,1\}^2,
	\]
	with \(\FC,\PD,\AD,\FI\) corresponding respectively to \((1,1),(1,0),(0,1),(0,0)\). Thus the spectrum is naturally isomorphic to the Belnap bilattice \(\mathcal{FOUR}\). The determination order
	\[
	\FI \leq_d \PD,\AD \leq_d \FC
	\]
	measures the amount of issue-determining power possessed by a coalition, whereas the directionality order
	\[
	\AD \leq_\delta \FI,\FC \leq_\delta \PD
	\]
	measures the orientation of that power toward \(\varphi\) or toward \(\neg\varphi\). This observation provides a concrete basis for bilattice-valued Coalition Logic, where effectivity information may be incomplete, inconsistent, or supplied by multiple conflicting sources.
	
	A related direction is to study effectivity as a polarity between coalitions and objectives. At each state \(w\), the relation
	\[
	C \Vdash_w X
	\quad\Longleftrightarrow\quad
	X\in E_w(C)
	\]
	connects the coalition lattice \((\powerset(N),\subseteq)\) with the outcome lattice \((\powerset(W),\subseteq)\). The induced Galois operators can identify strategically closed families of coalitions and objectives, linking Coalition Logic with formal concept analysis and algebraic theories of dependence.
	
	\subsection{Concluding Remarks}
	\label{subsec:concluding-remarks}
	
	Coalition Logic has traditionally centered on positive effectivity: what coalitions can enforce. Recent work \cite{wang2026logic} elevated inability to a first-class modality, establishing its structural properties and axiomatization. The present paper advances this program by revealing that inability itself has internal structure.
	
	The key insight is that inability to force $\varphi$ does not entail inability to force $\neg\varphi$. A coalition satisfying \emph{simple inability} $\neg\Eff{C}\varphi$ may still enforce $\neg\varphi$, retaining adversarial control. By contrast, a coalition satisfying \emph{Full Inability} $\FI_C(\varphi) \equiv \neg\Eff{C}\varphi \land \neg\Eff{C}\neg\varphi$ is unable to settle the issue either way. The truth of $\varphi$ is left to forces beyond the coalition's unilateral control.
	
	This distinction generates a four-fold strategic spectrum:
	\[
	\FC,\quad \PD,\quad \AD,\quad \FI.
	\]
	Within this spectrum, Full Inability is not a residual category but the cornerstone of a systematic algebraic and order-theoretic structure. Under $\alpha$-duality, the four categories exhibit Klein four-group symmetry. In playable models, they correspond to order-convex regions in the powerset lattice. The resulting framework reveals that coalitional inability is governed by algebraic symmetry, lattice-theoretic convexity, and conservative proof-theoretic extension.
	
	Thus, this paper contributes to the broader program of treating inability as a first-class object of logical study. Where earlier work established inability as an independent modality, the present work reveals its internal structure. The logic of inability remains a rich and fruitful direction for further development.

	\end{document}